\documentclass{article}
\usepackage[a4paper]{geometry}
\usepackage{enumerate}
\usepackage{amsmath}
\usepackage{amssymb,latexsym}
\usepackage{amsthm}
\usepackage{color}
\usepackage{cancel}
\usepackage{graphicx}
\usepackage{cite}
\usepackage{longtable}
\usepackage{amscd}
\usepackage{lineno}

\makeatletter
\usepackage{amssymb,amsmath,amsthm,latexsym}
\usepackage[mathscr]{eucal} \usepackage{mathrsfs}
\usepackage{color}
\usepackage{hyperref}
\usepackage{comment}
\let\wfs@comment@comment\comment
\let\comment\@undefined

\usepackage{changes}
\let\wfs@changes@comment\comment
\let\comment\@undefined

\newcommand\comment{%
    \ifthenelse{\equal{\@currenvir}{comment}}
    {\wfs@comment@comment}
    {\wfs@changes@comment}%
}
\definechangesauthor[name=Daniele, color=red]{DAN}
\definechangesauthor[name=Giuseppe, color=green]{GIUS}

\makeatother

\def\C{\mathcal C}

\def\x{\mathbf x}
\def\y{\mathbf y}

\def\mX{\mathcal X}

\def\Tr{{\rm Tr}}

\def\F{\mathbb F}

\def\ps@headings{
 \def\@oddhead{\footnotesize\rm\hfill\runningheadodd\hfill\thepage}
 \def\@evenhead{\footnotesize\rm\thepage\hfill\runningheadeven\hfill}
 \def\@oddfoot{}
 \def\@evenfoot{\@oddfoot}
}

\DeclareMathOperator{\rank}{rank}
\DeclareMathOperator{\rk}{rk}
\DeclareMathOperator{\GL}{GL}
\DeclareMathOperator{\PR}{PR}
\DeclareMathOperator{\LP}{LP}
\DeclareMathOperator{\WW}{W}
\DeclareMathOperator{\ZZ}{Z}
\DeclareMathOperator{\NN}{N}
\DeclareMathOperator{\End}{End}

\newcommand{\Fmk}{[n,k]_{q^m/q}}
\newcommand{\Fmkd}{[n,k,d]_{q^m/q}}

\usepackage{comment}

\usepackage{changes}
\begin{document}


\newtheorem{thm}{Theorem}[section]
\newtheorem{lemma}[thm]{Lemma}
\newtheorem{proposition}[thm]{Proposition}
\newtheorem{corollary}[thm]{Corollary}
\newtheorem{Property}[thm]{Property}
\newtheorem*{Main}{Main Theorem}
\newtheorem{Notation}[thm]{Notation}
\newtheorem{conj}[thm]{Conjecture}
\newtheorem{crit}[thm]{Criterion}
\newtheorem{oss}[thm]{Observation}
\newtheorem{question}[thm]{Question}

\newtheorem{definition}[thm]{Definition}
\newtheorem{rem}[thm]{Remark}
\newtheorem{ex}[thm]{Example}

\newtheorem*{main}{Main Theorem}

\newcommand{\fa}[1]{{\color{cyan}{#1}}}

\def\q{{^{q}}}
\def\qq{^{q^2}}
\def\qqq{{^{q^3}}}
\def\qqqq{{^{q^4}}}
\def\qqqqq{{^{q^5}}}
\def\F{\mathbb{F}}
\def\m{{m}}
\def\mq{{m^{q}}}
\def\mqq{{m^{q^2}}}
\def\mqqq{{m^{q^3}}}
\def\mqqqq{{m^{q^4}}}
\def\mqqqqq{{m^{q^5}}}
\def\eps{{\epsilon}}
\def\epsq{{\epsilon\q}}
\def\epsqq{{\epsilon\qq}}
\def\epsqqq{{\epsilon\qqq}}
\def\epsqqqq{{\epsilon\qqqq}}
\def\epsqqqqq{{\epsilon\qqqqq}}

\title{An infinite family of non-extendable MRD codes}


\author{
Daniele Bartoli\thanks{Dipartimento di Matematica e Informatica,
Universit\`a degli Studi di Perugia, Perugia, Italy.
\texttt{daniele.bartoli@unipg.it}}, 
Alessandro Giannoni\thanks{Dipartimento di Matematica e Applicazioni ``R.~Caccioppoli'',
Universit\`a di Napoli Federico II, Napoli, Italy.
\texttt{alessandro.giannoni@unina.it}, \texttt{giuseppe.marino@unina.it}, \texttt{alessandro.neri@unina.it}}, 
Giuseppe Marino\footnotemark[2], 
Alessandro Neri\footnotemark[2]
}
\date{}

\maketitle
\begin{abstract}
In the realm of rank-metric codes, Maximum Rank Distance (MRD) codes are optimal algebraic structures attaining the Singleton-like bound. A major open problem in this field is determining whether an MRD code can be extended to a longer one while preserving its optimality. 
This work investigates $\mathbb{F}_{q^m}$-linear MRD codes that are non-extendable but do not attain the maximum possible length. Geometrically, these correspond to scattered subspaces with respect to hyperplanes that are maximal with respect to inclusion but not of maximum dimension.
By exploiting this geometric connection, we introduce the first infinite family of non-extendable $[4,2,3]_{q^5/q}$ MRD codes. Furthermore, we prove that these codes are self-dual up to equivalence.
\end{abstract}

\section{Introduction}

Rank-metric codes, originally introduced by Delsarte \cite{delsarte1978bilinear} and Gabidulin \cite{gabidulin1985theory}, have gathered significant attention due to their applications in network coding, cryptography, and space-time coding. An $\mathbb{F}_{q^m}$-linear rank-metric code of length $n$, dimension $k$, and minimum rank distance $d$ attains the Singleton-like bound $d \leq n - k + 1$ (assuming $n \leq m$). Codes meeting this bound are called \emph{Maximum Rank Distance (MRD) codes}.

A fundamental problem in coding theory is the \emph{extension} of optimal codes. For $n < m$, the vast majority of known MRD codes are obtained by puncturing maximum-length codes (where $n=m$). An MRD code is called \emph{non-extendable} if it cannot be obtained by puncturing a longer MRD code. Equivalently, its generator matrix cannot be augmented with an additional column without strictly decreasing the minimum rank distance. Classifying non-extendable MRD codes is a difficult algebraic task, as it requires proving the non-existence of any valid extension.

To tackle this, we translate the problem into finite geometry using the theory of $q$-systems. As established in recent literature \cite{randrianarisoa2020geometric, zini2021scattered}, a nondegenerate $\mathbb{F}_{q^m}$-linear MRD code of dimension $k$ naturally corresponds to a $(k-1)$-scattered $\mathbb{F}_q$-subspace. Crucially, a non-extendable MRD code corresponds exactly to a \emph{maximally scattered subspace}: a scattered subspace that is not properly contained in any larger scattered subspace. 

While maximum scattered subspaces (corresponding to maximum-length MRD codes) have been extensively studied and classified in low dimensions, geometric objects that are less structured and more difficult to classify emerge when studying maximality without maximum dimension. Currently, only one sporadic example of a maximally scattered subspace not of maximum dimension is known \cite{blokhuis2000scattered, MR3587265}. 

In this paper, we construct the first infinite family of such objects. Our main result (Theorem \ref{MainTh}) provides an infinite family of non-extendable $[4,2,3]_{q^5/q}$ MRD codes for $q=3^{2h+1}$. We achieve this by analyzing $4$-dimensional $\mathbb{F}_q$-subspaces in $\mathbb{F}_{q^5}^2$ and exploiting the recent partial classification of maximum scattered subspaces in $\mathbb{F}_{q^5}^2$ obtained in \cite{Longobardi2025classificationmaximumscatteredlinear}. Finally, in Section \ref{Sec:DualCode}, we explore the dual properties of our construction, proving that these codes are self-dual up to equivalence.
\section{Preliminaries}\label{Sec:Connections}

\subsection{Scattered subspaces}
This section introduces the definitions and foundational results for scattered and evasive subspaces.

\begin{definition}
Let $k$ and $n$ be positive integers, and let $h$ and $r$ be non-negative integers satisfying $h < k$ and $h \le r$. An $\mathbb{F}_{q}$-subspace $U \subseteq \mathbb{F}_{q^m}^k:= \mathbb{F}_{q^m}^k$ is defined as \textbf{$(h,r)$-evasive} if, for any $h$-dimensional $\mathbb{F}_{q^m}$-subspace $H \subseteq \mathbb{F}_{q^m}^k$, the inequality $\dim_{\mathbb{F}_q}(U \cap H) \leq r$ holds.
In the specific case where $h=r$, an $(h,h)$-evasive subspace is called \textbf{$h$-scattered}. If $h=1$, a $1$-scattered subspace is simply referred to as \textbf{scattered}.
\end{definition}

The concept of scattered subspaces was first introduced by Blokhuis and Lavrauw in \cite{blokhuis2000scattered}. This idea was later extended to any integer $h$ by Csajbók, Marino, Polverino and Zullo in \cite{csajbok2021generalising}. The broader notion of evasive subspaces was formally defined in \cite{bartoli2021evasive}, though related concepts had appeared earlier in works such as \cite{pudlak2004pseudorandom,guruswami2011linear,dvir2012subspace,guruswami2016explicit}.

A well-established result for $h$-scattered subspaces is the upper bound on their $\mathbb{F}_{q}$-dimension. Specifically, for an $h$-scattered subspace $U \subseteq \mathbb{F}_{q^m}^k$, we have:
\begin{equation}\label{eq:scattered_bound}
\dim_{\mathbb{F}_q}(U) \leq \frac{km}{h+1},
\end{equation}
as shown in \cite{blokhuis2000scattered,csajbok2021generalising}. 

\begin{definition}
  An $h$-scattered subspace that achieves Bound \eqref{eq:scattered_bound} is called a \textbf{maximum $h$-scattered} subspace (or \textbf{maximum scattered} subspace, when $h=1$). An $h$-scattered subspace which is  not properly contained within any larger $h$-scattered subspace, is called a  \textbf{maximally $h$-scattered} subspace (or \textbf{maximally scattered} subspace, when $h=1$).
\end{definition}
 Consequently, the dimension (or rank) of a scattered subspace in $\mathbb{F}_{q^m}^k$ is at most $km/2$. If this maximum is met, the subspace is known as a maximum scattered subspace.


The first example of a maximally scattered subspace which is not maximum scattered -- a maximally scattered $5$-dimensional subspace of $\mathbb{F}_{2^6}^2$ -- was given in \cite[Example 3.2]{MR3587265} (see also \cite{blokhuis2000scattered}) and was  constructed using the GAP-package FinInG.  



So far,  no infinite families of maximally scattered linear sets which are not maximum scattered  are known and the aim of this paper is to provide the first infinite family of such objects.

On maximally scattered subspaces, the following result provides a useful lower bound on its dimension; see \cite[Theorem 2.1]{blokhuis2000scattered}.
\begin{thm}\label{LowerBound} 
    If $U\leq \mathbb{F}_{q^m}^k$ is a maximally scattered subspace, then 
    $$\textrm{dim}_{\mathbb{F}_q} (U) \geq \left \lceil \frac{km-m}{2}\right \rceil+1.$$
\end{thm}

The above theorem can be generalized to the case of $h$-scattered subspaces.
\begin{thm}
    If $U\leq \mathbb{F}_{q^m}^k$ is a maximally $h$-scattered subspace, then

 $$\textrm{dim}_{\mathbb{F}_q} (U) \geq \left\lceil \frac{km-hm}{h+1}\right\rceil+h.$$

\end{thm}
\begin{proof}
    Let $s$ be the dimension of an $h$-scattered subspace $U$ in $\mathbb{F}_{q^m}^k$. The subspace $U$ can be extended to a larger $h$-scattered subspace $U'$ if and only if there exists a vector $v\in \mathbb{F}_{q^m}^k$ not contained in any space generated by $U$ and an $\mathbb{F}_{q^m}$-subspace of  $\mathbb{F}_{q^m}^k$ whose intersection with $U$ has dimension $h$ (i.e., of weight $h$ in $U$). First we have to determine an upper bound on the number of such $\mathbb{F}_{q^m}$-subspaces. The number of $h$-dimensional $\F_{q^m}$-subspaces of weight $h$ in $U$ is precisely the number of $h$-dimensional $\F_q$-subspaces of $U$, which is $\binom{s}{h}_q$.  The total number of vectors of $\mathbb{F}_{q^m}^k$ contained in any of the spaces generated by $U$ and an $h$-dimensional $\mathbb{F}_{q^m}$-subspace of weight $h$ is at most
    $$(q^{hm+s-h}-q^s) \binom{s}{h}_q+q^s.$$
    Thus, $U$ can be extended to a larger $h$-scattered subspace if 
    $q^{km}>(q^{hm+s-h}-q^s) \binom{s}{h}_q+q^s$
    and in particular if
    $ km> hm+s-h+h(s-h),$
    i.e., 

    $$s<\frac{km-hm}{h+1}+h. $$
    This is equivalent to say that if $U$ is maximally scattered then

    $$s\geq \left\lceil \frac{km-hm}{h+1}\right\rceil+h.$$
\end{proof}

\subsection{Linearized polynomials and Dickson matrices}

A key tool for describing $\mathbb{F}_q$-subspaces of $\mathbb{F}_{q^m}^k$ is the class of linearized polynomials (or $q$-polynomials) over $\mathbb{F}_{q^m}$. Let $\tilde{\mathcal{L}}_{m,q}$ denote the $\mathbb{F}_q$-algebra of linearized polynomials over $\mathbb{F}_{q^m}$ reduced modulo $x^{q^m}-x$. Every polynomial in $\tilde{\mathcal{L}}_{m,q}$ has the form 
$$L(x) = \sum_{i=0}^{m-1} a_i x^{q^i} \in \mathbb{F}_{q^m}[x].$$ 

It is a well-known fact that $\tilde{\mathcal{L}}_{m,q}$ is isomorphic to the $\mathbb{F}_q$-algebra of $\mathbb{F}_q$-linear endomorphisms of $\mathbb{F}_{q^m}$, namely $\End_{\mathbb{F}_q}(\mathbb{F}_{q^m})$. Consequently, for any $L \in \tilde{\mathcal{L}}_{m,q}$, its kernel is a well-defined $\mathbb{F}_q$-subspace of $\mathbb{F}_{q^m}$, given by
$$ \ker(L) := \{ x \in \mathbb{F}_{q^m} : L(x) = 0 \}. $$
Similarly, the rank of $L(x)$ is defined as the $\mathbb{F}_q$-dimension of its image. By the Rank-Nullity Theorem, we have
$$ \rank(L) := m - \dim_{\mathbb{F}_q}(\ker(L)). $$

The rank of such an operator is strictly related to the properties of its associated Dickson matrix.

\begin{definition}
Let $L(x) = \sum_{i=0}^{m-1} a_i x^{q^i} \in \tilde{\mathcal{L}}_{m,q}$. The \textbf{Dickson matrix} associated with $L(x)$ is the $m \times m$ matrix 
\[
D(L) = \begin{pmatrix}
a_0 & a_1 & \dots & a_{m-1} \\
a_{m-1}^q & a_0^q & \dots & a_{m-2}^q \\
\vdots & \vdots & \ddots & \vdots \\
a_1^{q^{m-1}} & a_2^{q^{m-1}} & \dots & a_0^{q^{m-1}}
\end{pmatrix}.
\]
\end{definition}

It is a well-known fact that the rank of $L(x)$ equals the rank of $D(L)$ over $\mathbb{F}_{q^m}$. This algebraic characterization allows us to translate geometric conditions on scattered subspaces into matrix rank conditions.

\subsection{Rank metric codes}
It is worth noting that linear $\F_q$-subspaces of $\mathbb{F}_{q^m}^k$ have a geometric connection to rank metric codes with some additional linear property. We start by introducing the theory of rank metric codes. 

\begin{definition}
    A \textbf{rank metric code} is a subset $\mathcal{C} \subset \F_{q}^{r \times s}$ of matrices, equipped with the \textbf{rank distance} 
    $$d(M,N) = \rank(M-N) \qquad \mbox{for any } M,N \in \C. $$
    The \textbf{minimum rank distance} of $\C$ is the integer
      \[
d_{\rk}(\C):=\min \{d_{\rk}(M,N) : M,N \in \mathcal{C}, M\neq N\}.
\]
    If in addition $\C$ is an $\F_q$-subspace, then $\C$ is said to be \textbf{$\F_q$-linear}, and in this case 
    \[
d_{\rk}(\C)=\min \{\rank(M) : M \in \mathcal{C} \setminus \{{\bf 0}\}\}.
\]
\end{definition} 

The parameters of a rank-metric code $\C$ are constrained by the Singleton-like bound:
\[
|\mathcal{C}| \leq q^{\max\{r,s\} (\min\{r,s\} - d_{\rk}(\C) + 1)},
\]
as established by Delsarte in \cite{delsarte1978bilinear}. Codes that attain this bound are known as \textbf{maximum rank distance (MRD) codes}, and they represent the most studied family of rank-metric codes. We will see later the most prominent family of MRD codes, which was given independently by Delsarte \cite{delsarte1978bilinear} and Gabidulin \cite{gabidulin1985theory}. 
These constructions of MRD codes exploit an additional notion of linearity, which we now introduce.
Recall that a field $\F_{q^m}$ is isomorphic to $\F_q^m$ as $\F_q$-vector space and any isomorphism can be derived by choosing an ordered $\F_q$-basis $\mathcal B$ of $\F_{q^m}$. In the same way, the space $\F_{q^m}^n$ is isomorphic to $\F_q^{m\times n}$ as $\F_q$-vector space. 

Hence, one can define the rank distance on $\F_{q^m}^n$ as the one induced by the rank distance in $\F_q^{m\times n}$ after writing every entry in coordinates with respect to an $\F_q$-basis $\mathcal B$ of $\F_{q^m}$. Since using a different $\F_q$-basis acts as an invertible matrix on $\F_q^{m\times n}$ -- that is, the change of basis matrix -- the rank distance does not depend on the choice of the $\F_q$-basis $\mathcal B$. Actually, it is well-known that the rank distance can be intrinsically defined on $\F_{q^m}^n$ as follows:
$$d_{\rk}(u,v)=\dim_{\F_q}(\langle u_1-v_1,\ldots,u_n-v_n\rangle_{\F_q}), \qquad \mbox{ for  } u=(u_1,\ldots,u_n), v=(v_1,\ldots, v_n) \in \F_{q^m}^n.$$
In this way, one has a natural notion of $\F_{q^m}$-linearity that is hidden in the matrix representation. From now on, we will only consider codes possessing such a linearity.

\begin{definition}
    An $\Fmkd$ \textbf{(rank-metric) code} $\C$ is a $k$-dimensional $\F_{q^m}$-subspace of $\F_{q^m}^n$, endowed with the rank distance, where $d=d_{\rk}(\C)$ is the \textbf{minimum rank distance} of $\C$. If the minimum rank distance of $\C$ is not known nor relevant, we will simply refer to $\C$ as an $\Fmk$ code.
\end{definition}

A \textbf{generator matrix} for an $[n,k]_{q^m/q}$ code is a matrix $G \in \F_{q^m}^{k\times n}$ such that $$\C=\mathrm{rowsp}(G)\leq \F_{q^m}^n,$$
where $\mathrm{rowsp}(\cdot)$ denotes the $\F_{q^m}$-vector space generated by the rows. If the columns of one (and hence all) generator matrix $G$ of $\C$ are $\F_q$-linearly independent, then the code is said to be \textbf{nondegenerate}. The terminology appears evident because this property on the generator matrix is equivalent to say that $\C$ cannot be isometrically embedded in a smaller ambient space $(\F_{q^m}^{n'},d_{\rk})$, with $n'<n$; see e.g. \cite{ABNR22}.

\medskip

There is a strong connection between $[n,k]_{q^m/q}$ nondegenerate rank-metric codes and $n$-dimensional $\F_q$-subspaces of $\F_{q^m}^k$ that are not contained in any $\F_{q^m}$-hyperplane. These $\F_q$-subspaces are nowadays known as \textbf{$q$-systems}, and their connection to rank-metric codes was explicitly introduced for the first time in \cite{randrianarisoa2020geometric}. Given an $n$-dimensional $\F_q$-subspace $U$ of $\F_{q^m}^k$ with the property that $\langle{U}\rangle_{\mathbb{F}_{q^m}}=\F_{q^m}^k$ -- namely, a $q$-system -- and given an $\mathbb{F}_{q}$-basis $\{g_1, \dots, g_n\}$ of $U$, one can construct a matrix $G\in \F_{q^m}^{k\times n}$ whose columns are these basis vectors, that is,
$$G=\begin{pmatrix}
    | & & | \\
    g_1 & \cdots  & g_n\\
    | & & |
\end{pmatrix}.$$ This matrix $G$ can be viewed as the generator of an $[n,k]_{q^m/q}$ nondegenerate rank-metric code $\C$. In other words, 
$$\C=\mathrm{rowsp}(G)\leq \F_{q^m}^n.$$

Vice-versa, if $\C$ is an $\Fmk$ nondegenerate code, then for a given generator matrix $G\in \F_{q^m}^{k\times n}$, we can select its columns and consider the $\F_q$-subspace $U$ that they generate inside $\F_{q^m}^k$. By the fact that $\C$ is nondegenerate, the $\F_q$-dimension of $U$ is $n$. Moreover, since the rank of $G$ is $k$, then $\langle U \rangle_{\F_{q^m}}=\F_{q^m}^k$, implying that $U$ is indeed a $q$-system.

The two procedures described above are not well-defined, since the first depends on the choice of the basis $\{g_1,\ldots,g_n\}$ of $U$, while the second depends on the choice of the generator matrix $G$ of $\C$ (or, equivalently, on the choice of a basis of $\C$). However, if we consider $\F_q$-subspaces of $\F_{q^m}^k$ up to $\GL(k,q^m)$-equivalence, and the $[n,k]_{q^m/q}$ codes up to $\GL(n,q)$-equivalence, the two procedures are one the inverse of the other, and we obtain a bijection between equivalence classes of $[n,k]_{q^m/q}$  {nondegenerate} codes, and equivalence classes of $\F_q$-subspaces of $\F_{q^m}^k$ of $\F_q$-dimension $n$; see \cite{randrianarisoa2020geometric,ABNR22}.

In the remainder of the paper, we will say that a $q$-system $U$ is an $\Fmkd$ \textbf{system associated} with a nondegenerate $\Fmkd$ code $\C$ if it can be obtained from $\C$ by the above procedure on one of its generator matrices. Similarly, we will say that $\C$ is associated with $U$.




The importance of $q$-system is given by the fact that they represent the geometrical counterparts of $\Fmk$ codes, and their scatteredness properties
reveal the property of the associated codes of being MRD. This was shown in a series of papers \cite{MR4079480,zini2021scattered}. We reformulate these findings in the following more compact way.

\begin{thm}[see {\cite[Corollary 5.7]{MR4079480}, \cite[Theorem 3.2]{zini2021scattered}}]\label{thm:MRD-iff-scattered}
 Let $\C$ be a nondegenerate $\Fmk$ code, and let $U$ be any of its associated $\Fmk$ systems. Then, $\C$ is MRD if and only if one of the following holds:
 \begin{enumerate}
     \item $n\le m$ and $U$ is $(k-1)$-scattered, or
     \item $n>m$, $n$ divides $km$ and $U$ is maximum $h$-scattered, with $h=\frac{km}{n}-1$.
 \end{enumerate}
\end{thm}

For a full overview of the relation between scattered and evasive $q$-system and the parameters of their associated rank-metric codes, we refer the interested reader to \cite{MR4621445}. 

\medskip 

MRD codes have been studied since the late 70's, when Delsarte provided their first construction \cite{delsarte1978bilinear}. This construction was rediscovered a few years later by Gabidulin \cite{gabidulin1985theory}. These codes are now known as Delsarte-Gabidulin codes and can be described in the following way. 

Let $v=(v_1,\ldots.v_n)\in \F_{q^m}^n$ be such that $\dim_{\F_q}(\langle v_1,\ldots,v_n\rangle_{\F_q})=n$, and let $k\le n$. The \textbf{Delsarte-Gabidulin code} $\C_k(v)$ is the $[n,k]_{q^m/q}$ code whose generator matrix is given by the $k\times n$ \textbf{Moore matrix}
$$M_k(v)=\begin{pmatrix}
    v_1 & v_2 & \cdots & v_n \\
    v_1^q & v_2^q & \cdots & v_n^q \\
    \vdots & \vdots & & \vdots \\
    v_1^{q^{k-1}} & v_2^{q^{k-1}} & \cdots & v_n^{q^{k-1}}
\end{pmatrix}.$$

Using the procedure described previously, it is easy to see that one of its associated $\Fmk$ systems is 
$$V_k(\mX):=\{(x,x^q,\ldots,x^{q^{k-1}}\,:\, x \in \mX\},$$
where $\mX=\langle v_1,\ldots, v_n\rangle_{\F_q}$. 

Apart from Delsarte-Gabidulin codes, only few other systematic constructions of $\Fmk$ MRD codes are known; see e.g. \cite{sheekey2016new}. All such constructions have an algebraic flavor, and rely on the representation of the space of square matrices $\F_q^{m\times m}$ as $\End_{\F_q}(\F_{q^m})$. Hence, these constructions are all designed for the square case. For what concerns the rectangular case, especially for $n<m$, the known examples of $[n,k,n-k+1]_{q^m/q}$ MRD are obtained from a square $[m,k,m-k+1]_{q^m/q}$ MRD code by \textbf{puncturing}. 

\begin{definition}
   Let $\C$ be an $\Fmk$ code and let $A\in \F_q^{n\times n'}$ be such that $\rank(A)=n'$. The \textbf{puncturing} of $\C$ on $A$ is the $[n',k']_{q^m/q}$ code 
   $$\pi(\C,A):=\{cA\,:\, c \in \C\}.$$
\end{definition}

In terms of associated $q$-systems, the following result is straightforward.

\begin{thm}\label{thm:puncturing_rankmetric}
    Let $\C$ be an $\Fmk$ code, and let $U$ be any of its associated $\Fmk$ systems. Then, an $[n',k]_{q^m/q}$ code $\C'$ is a puncturing of $\C$ if and only if there exists an $[n',k]_{q^m/q}$ system $U'$ associated with $\C'$ such that $U'\leq U$.
\end{thm}

\begin{proof}
    First, suppose that $U'$ is an $[n',k]_{q^m/q}$ system contained in $U$. Let $\{g_1,\ldots,g_{n'}\}$ be an $\F_q$-basis of $U'$, and let $G'\in \F_{q^m}^{k\times n'}$  be the matrix whose $i$th column is $g_i$. Clearly, we have that $\C'=\mathrm{rowsp}(G')$ is an $[n',k']_{q^m/q}$ code associated with $U'$.  Let us complete  $\{g_1,\ldots,g_{n'}\}$ to an $\F_q$-basis $\{g_1,\ldots,g_{n}\}$ of $U$ and let $G\in \F_{q^m}^{k\times n}$  be the matrix whose $i$th column is $g_i$. for $i \in \{1,\ldots,n\}$. In this way, $\C=\mathrm{rowsp}(G)$. Then, if we consider the matrix $A=(I_{n'}\,\mid\,0 )^\top$, we have $G'=GA$, and hence $$\C'=\mathrm{rowsp}(G')=\mathrm{rowsp}(GA)=\pi(\C,A).$$

    On the other hand, let us assume that $\C'$ is a puncturing of $\C$. Then, there exists $A\in \F_q^{n\times n'}$ of rank $n'$ such that $\C'=\pi(\C,A)$. Hence, if $G$ is a generator matrix of $\C$, then $G'=GA$ is a generator matrix of $\C'$. By definition of associated $q$-system, this means that, there exists a generator matrix $G$ of $\C$ such that $U$ is the $\F_q$-span of the columns of $G$. Moreover, the $\F_q$-span $U'$ of the columns of $G'$ is clearly an $\F_q$ subspace of $U$, concluding the proof.
\end{proof}

As mentioned above, all the known constructions of $\Fmk$ MRD codes with $n<m$ are obtained by puncturing an $[m,k]_{q^m/q}$ MRD code. It is natural to ask whether one can construct new $\Fmk$ MRD codes that cannot be obtained by puncturing a longer MRD code. In terms of the associated $q$-systems, by Theorem \ref{thm:puncturing_rankmetric}, we can characterize such codes as follows.

\begin{corollary}
    Let $\C$ be an $\Fmk$ code with $n\le m$, and let $U$ be any of its associated $\Fmk$ systems. Then, $\C$ is an MRD codes that cannot be obtained by puncturing a $[n+1,k]_{q^m/q}$ code if and only if $U$ is maximally $(k-1)$-scattered.
\end{corollary}

\begin{proof}
    This directly follows from Theorem \ref{thm:MRD-iff-scattered} and Theorem \ref{thm:puncturing_rankmetric}.
\end{proof}

Thus, constructing maximally $(k-1)$-scattered $\Fmk$ systems that are not maximum coincides with constructing $\Fmk$ MRD codes with $n<m$ that are not obtained by puncturing a longer MRD code. We will call such codes \textbf{non-extendable}, following the terminology of the analogous codes in the classical Hamming metric. This represents the main motivation of our work.

\medskip




\subsection{A few things on algebraic varieties}

In this section, we summarize several concepts and results related to algebraic varieties.
We use the notations $\mathbb{P}^r(\mathbb{F}_q)$ and $\mathbb{A}^{r}(\mathbb{F}_{q})$ (or simply $\mathbb{F}_{q}^r$) to denote the projective and the affine space of dimension $r\in \mathbb{N}$ over the finite field $\mathbb{F}_{q}$, respectively. Let $\overline{\mathbb{F}_q}$ denote the algebraic closure of $\mathbb{F}_q$.

A variety $\mathcal{V}$ is defined as the set of common zeros of a finite collection of polynomials. Specifically, an affine $\mathbb{F}_q$-rational variety (or an affine variety defined over $\mathbb{F}_q$) is a set $\mathcal{V}\subset \mathbb{A}^r(\overline{\mathbb{F}_q})$ for which there exist polynomials $F_1, \ldots, F_s$ in the polynomial ring $\mathbb{F}_q[X_1, \ldots, X_r]$ such that:
\[
\mathcal{V} = \{ (a_1, \ldots, a_r) \in \mathbb{A}^r(\overline{\mathbb{F}_q}) \mid F_i(a_1, \ldots, a_r) = 0 \text{ for all } i = 1, \ldots, s \}.
\]
This set is also denoted as $V(F_1, \ldots, F_s)$.
Similarly, a  projective $\mathbb{F}_q$-rational variety (or a projective variety defined over $\mathbb{F}_q$) in $\mathbb{P}^r(\overline{\mathbb{F}_q})$ is defined using polynomials $F_1, \ldots, F_s \in \mathbb{F}_q[X_0, X_1, \ldots, X_r]$, with the additional requirement that each polynomial $F_i$ must be homogeneous.
The set of $\mathbb{F}_q$-rational points of an $\mathbb{F}_q$-rational variety $\mathcal{V}$ is given by the intersection $\mathcal{V}\cap \mathbb{A}^r(\mathbb{F}_q)$ or $\mathcal{V}\cap \mathbb{P}^r(\mathbb{F}_q)$, and it is usually denoted by $\mathcal{V}(\mathbb{F}_q)$.
A hypersurface is a variety defined by a single polynomial.
We say that a variety $\mathcal{V}$ is \textit{absolutely irreducible} if it cannot be expressed as the union of two proper subvarieties defined over the algebraic closure $\overline{\mathbb{F}_q}$. That is, there are no varieties $\mathcal{V}'$ and $\mathcal{V}''$ defined over $\overline{\mathbb{F}_q}$ and different from $\mathcal{V}$ such that $\mathcal{V}=\mathcal{V}' \cup \mathcal{V}''$.
In the specific case of a hypersurface $\mathcal{V}=V(F)\subset\mathbb{A}^r(\overline{\mathbb{F}_q})$ (or projective space), it is \textit{absolutely irreducible} if and only if its defining polynomial $F$ is irreducible over the algebraic closure. That is, there are no non-constant polynomials $G,H\in \overline{\mathbb{F}_q}[X_1,\ldots,X_r]$ (or homogeneous polynomials in $\overline{\mathbb{F}_q}[X_0,\ldots,X_r]$) such that $F=GH$.
The dimension of a variety can be defined as the maximal integer $s$ for which there exists a chain of distinct, nonempty, absolutely irreducible varieties contained in $\mathcal{V}$:
\[
\emptyset = \mathcal{V}_0 \subsetneq \mathcal{V}_1 \subsetneq \cdots \subsetneq \mathcal{V}_{s+1} = \mathcal{V}.
\]
 We say that an $s$-dimensional projective variety $\mathcal{V}$ has degree $d$, written $\deg(\mathcal{V})=d$, if $d$ is the number of intersection points with a general projective subspace of complementary dimension. That is,
\[
d = \#(\mathcal{V}\cap H),
\]
where $H \subseteq \mathbb{P}^r(\overline{\mathbb{F}_q})$ is a general projective subspace of dimension $r-s$. Algebraic varieties $\mathcal{V}\subset \mathbb{A}^r(\overline{\mathbb{F}_q})$ (or $\mathcal{V}\subset \mathbb{P}^r(\overline{\mathbb{F}_q})$) of dimension $1$, $2$, and $r-1$ are called curves, surfaces, and hypersurfaces, respectively.

Determining the degree of a variety is generally not straightforward; however, an upper bound to $\deg(\mathcal{V)}$ is given by $\prod_{i=1}^{s}\deg(F_i).$ 
We also recall that the Frobenius map $\Phi_q: x \mapsto x^q$ is an automorphism of $\mathbb{F}_{q^k}$ and generates the group $Gal(\mathbb{F}_{q^k} / \mathbb{F}_q)$ of automorphisms of  $\mathbb{F}_{q^k}$ that fixes $\mathbb{F}_{q}$ pointwise.
The Frobenius automorphism also induces  a collineation of $\mathbb{A}^r(\overline{\mathbb{F}_q})$ and an automorphism of $\overline{\mathbb{F}_{q}}[X_1,\dots,X_r].$ 

From now on, with a slight abuse of notation, we will write $\mathcal{V}\subset \mathbb{A}^r(\mathbb{F}_q)$ or $\mathcal{V}\subset \mathbb{P}^r(\mathbb{F}_q)$ to indicate that the variety $\mathcal{V}$ is $\mathbb{F}_q$-rational.



\section{An infinite family of maximally scattered subspaces in $\mathbb{F}_{q^5}^2$}

This section presents the paper's main result, the first known infinite family of maximally scattered subspaces.

Let start this section by showing a non-existence result of maximally $(k-1)$-scattered $q$-systems in a special case, by using their correspondence with MRD codes.



\begin{proposition}\label{Prop:condNec}
    Let $q$ be a prime power and let $k<m$.  The only maximally $(k-1)$-scattered $[k+1,k]_{q^m/q}$ systems are all equivalent to 
    $$V_k(\mX)=\{(x,x^q,\ldots,x^{q^{k-1}})\,:\, x \in \mX\},$$
    for some $\F_q$-subspace $\mX\leq \F_{q^m}$ of dimension $k+1$. In particular, if $k+2\le m$, then there are no maximally $(k-1)$-scattered $[k+1,k]_{q^m/q}$ systems.
\end{proposition}

\begin{proof}
    By Theorem \ref{thm:MRD-iff-scattered}, if $U$ is
    a $(k-1)$-scattered $[k+1,k]_{q^m/q}$ system, then any of its associated codes must be a $[k+1,k,2]_{q^m/q}$ MRD code $\C$. Moreover, by \eqref{eq:scattered_bound}, this implies that $k+1\le m$. It is well known that every MRD code of codimension $1$ is a Delsarte-Gabidulin code \cite{gabidulin1985theory}, which means that $U$ is equivalent to $\{(x,x^q,\ldots,x^{q^{k-1}})\,:\, x \in \mX\}$, for some $\F_q$-subspace $\mX\leq \F_{q^m}$ of dimension $k+1$.
    
    In addition, if $k+1\le m-1$, up to equivalence, $U$ is properly contained in  the $(k-1)$-scattered $[m,k]_{q^m/q}$ system $\{(x,x^q,\ldots,x^{q^{k-1}})\,:\, x \in \F_{q^m}\}$,
    which has dimension $m>k+1$. Thus, $U$ cannot be maximally $(k-1)$-scattered.
\end{proof}

We focus on the case $k=2$. The following result shows that every maximally scattered subspace in $\mathbb{F}_{q^m}^2$ is also maximum scattered, when $m\leq 4$. 

\begin{proposition}
Let $q$ be a prime power. For $m \leq 4$, any maximally scattered subspace in $\mathbb{F}_{q^m}^2$ is maximum.
\end{proposition}
\begin{proof}
By Theorem \ref{LowerBound}, if $m\leq 3$ then a maximally scattered subspace has dimension at least $m$, which is the dimension of a maximum scattered subspace. 

Let us consider the case $m=4$. By Theorem \ref{LowerBound}, a maximally scattered subspace in $\mathbb{F}_{q^4}^2$ has dimension at least $3$. By Proposition \ref{Prop:condNec} a maximally $1$-scattered subspace with these parameters cannot exist.
\end{proof}

Thus, for our purposes, we focus on the first open case, that is $\mathbb{F}_{q^5}^2$. By Theorem \ref{LowerBound}, a maximally scattered subspace in $\mathbb{F}_{q^5}^2$ has dimension 4. 

Let $\NN_{q^5/q}$ and $\Tr_{q^5/q}$ denote the norm and the trace functions from $\F_{q^5}$ to $\F_q$, respectively.

Our investigation exploits the partial classification of maximum scattered subspaces in \cite{Longobardi2025classificationmaximumscatteredlinear}. 

\begin{thm}[{see \cite[Section 6]{Longobardi2025classificationmaximumscatteredlinear}}]\label{Th:Classification}
Any maximum scattered subspace in $\mathbb{F}_{q^5}^2$
is up to equivalence in $\Gamma \mathrm{L}(2,q^5)$ one of the following
   \begin{enumerate}
       \item[(C1)]  \qquad $\PR_s:=\{(x,x^{q^s}) :x\in\mathbb{F}_{q^5}\}$, $s\in \{1,\ldots,4\}$;

\item[(C2)] \qquad $\LP_{s,\eta}:=\{(x,x^{q^s}+\eta x^{q^{5-s}}):x\in\mathbb{F}_{q^5}\}$, $s\in\{1,2\}$, 
$\NN_{q^5/q}(\eta)\ne0,1$;

\item[(C3)] \qquad $\WW_{\eta,\rho}:=\{(\eta(x^q-x)+\Tr_{q^5/q}(\rho x),x^q-x^{q^4}):x\in\mathbb{F}_{q^5}\}$,
$\eta\ne0$, $\Tr_{q^5/q}(\eta)=0\ne \Tr_{q^5/q}(\rho)$;

\item[(C4)] \qquad $\ZZ_k := \{(x,k(x^q+x^{q^3})+x^{q^2}+x^{q^4}) :x\in\mathbb{F}_{q^5}\}$, 
$\NN_{q^5/q}(k)=1$.
\end{enumerate}
The classes of sets of types (C3) and (C4) might be empty, as they actually are for $q\le25$.

\end{thm}

The purpose of the rest of this section is to provide the first infinite family of maximally scattered subspaces in $\mathbb{F}_{q^5}^2$ which are not maximum scattered. 

Currently, the family is proven to be maximally scattered only for finite fields of order $q=3^{2h+1}$ (where $h \geq 1$) and for specific choices of the parameter $\delta$. The proof relies on the partial classification of maximum scattered subspaces in $\mathbb{F}_{q^5}^2$ from \cite{Longobardi2025classificationmaximumscatteredlinear}. While the formal proof is restricted to this case, our experimental results suggest that the construction holds in any characteristic.

Notably, generalizing this result primarily depends on proving that the families designated as (C3) and (C4) in the aforementioned classification are empty.
\medskip

Let $q$ be an odd prime power and let $\mX \le \mathbb{F}_{q^5}$ be an $\mathbb{F}_q$-hyperplane. Every such hyperplane can be described as $\mX_\lambda := \{ x \in \mathbb{F}_{q^5} : \Tr_{q^5/q}(\lambda x) = 0 \}$ for some $\lambda \in \mathbb{F}_{q^5}^*$. 

Consider $\delta \in \mathbb{F}_{q^5}$ satisfying $\NN_{q^5/q}(\delta)=1$. We define the corresponding family of subspaces as
$$U_{\delta}(\mX_\lambda) := \{(x,x^{q}+\delta x^{q^4}) : x \in \mX_\lambda\} \leq \mathbb{F}_{q^5}^2.$$

\begin{proposition}\label{U:equiv}
    The subspace $U_\delta(\mX_\lambda)$ is equivalent via the diagonal matrix $D = \mathrm{diag}(\lambda, \lambda^q) \in \mathrm{GL}(2,q^5)$ to the subspace 
    $$U_{\delta'}(\mX_1)= \{(y,y^q+\delta' y^{q^4}) : \Tr_{q^5/q}(y)=0\},$$
    where $\delta' = \delta \lambda^{q-q^4}$. Furthermore, $\NN_{q^5/q}(\delta') = \NN_{q^5/q}(\delta)=1$.
\end{proposition}
\begin{proof}

Applying the diagonal matrix $D = \mathrm{diag}(\lambda, \lambda^q)$ to a generic element of $U_{\delta}(\mX_\lambda)$, we obtain
$$D \begin{pmatrix} x \\ x^q + \delta x^{q^4} \end{pmatrix} = \begin{pmatrix} \lambda x \\ \lambda^q x^{q} + \lambda^q \delta x^{q^4} \end{pmatrix}.$$
By setting $y = \lambda x$, the condition $x \in \mX_\lambda$ becomes $y \in \mX_1$, that is, $\Tr_{q^5/q}(y) = 0$. The second coordinate translates to $y^q + \lambda^q \delta (\lambda^{-q^4} y^{q^4}) = y^q + \delta' y^{q^4}$.

This shows $D(U_{\delta}(\mX_\lambda)) = U_{\delta'}(\mX_1)$.
Since $\NN_{q^5/q}(\lambda) \in \mathbb{F}_q$, we have $\NN_{q^5/q}(\lambda)^{q-q^4} = 1$. Therefore, we have $\NN_{q^5/q}(\delta') = \NN_{q^5/q}(\delta) \NN_{q^5/q}(\lambda)^{q-q^4} = \NN_{q^5/q}(\delta)=1$.

\end{proof}

For this reason, up to equivalence via $\mathrm{diag}(\lambda, \lambda^q)$, we can restrict our analysis to the specific hyperplane $\mX_1 = \ker(\Tr_{q^5/q})$. From now on, we simply write
$$U_{\delta} := \{(x,x^{q}+\delta x^{q^4}) : \Tr_{q^5/q}(x)=0\}, \quad \NN_{q^5/q}(\delta)=1.$$

We first classify those $\delta$ such that $U_\delta$ is scattered, and to do this the following proposition will be helpful.

\begin{proposition}\label{incomp}
    There is no $\delta\in\F_{q^5}$ such that $\NN_{q^5/q}(\delta)=1$ and $\delta^{q^2+1}-\delta+1=0$.
\end{proposition}
\begin{proof}
    Consider $\delta^{q^2}=(\delta-1)/\delta$, and its $q^i$-th power with $i=1,2$
    $$\begin{cases}
        \delta^{q^2}=(\delta-1)/\delta\\
        \delta^{q^3}=(\delta^q-1)/\delta^q\\
        \delta^{q^4}=(\delta^{q^2}-1)/\delta^{q^2}.
    \end{cases}$$
    Using $\NN_{q^5/q}(\delta)=1$ the last equation reads $\delta^{-1-q-q^2-q^3}=(\delta^{q^2}-1)/\delta^{q^2}$ that is equivalent to $1+\delta^{q^3}\delta^q\delta(-\delta^{q^2}+1)=0$ and substituting the information we have on $\delta^{q^2},\delta^{q^3}$ by the first and second equation above we obtain 
$$0=1+\frac{\delta^q-1}{\delta^q}\delta^q\delta\left(-\frac{\delta-1}{\delta}+1\right)=1+(\delta^q-1)(-\delta+1+\delta)=\delta^q,$$
a contradiction to $\NN_{q^5/q}(\delta)=1$.
\end{proof}
Let $f_\delta(Y)=Y^2+(\delta^{q^3+q^2+q+1} - \delta^{q^2+q+1} +\delta^{q+1} -\delta - 1)Y-\delta^{q^3+q^2+q+2} -\delta^{q+1}+\delta\in \mathbb{F}_{q^5}[Y]$.

\begin{proposition}\label{Prop:Cond}
 The subspace  $U_{\delta}$ is a scattered subspace of dimension $4$ if and only if $f_\delta(Y)$ has no roots in $\F_{q^5}$.
\end{proposition}
\begin{proof}
    To prove that $U_\delta$ is scattered we need to prove that $\rank(mx+x^q+\delta x^{q^4})\geq 4$ for every $m\in\F_{q^5}$.

Let $g_m(x)=mx+x^q+\delta x^{q^4}=mx+x^q+\delta(-x-x^q-x^{q^2}-x^{q^3})$.
Consider the system

$$ g_m(x)=(g_m(x))^q=(g_m(x))^{q^2}=(g_m(x))^{q^3}=(g_m(x))^{q^4}=0. $$

It can be written as 
$$\begin{pmatrix}
    m - \delta &   -\delta + 1 &  -\delta & -\delta\\
    \delta^q &  m^q &  1 &  0\\
    0 &  \delta^{q^2} &  m^{q^2} &  1\\
    -1  & -1 &  \delta^{q^3} - 1 &  m^{q^3} - 1\\
    -m^{q^4} + 1 &  -m^{q^4}  & -m^{q^4}   &-m^{q^4} +\delta^{-q^3-q^2-q-1}
\end{pmatrix}\begin{pmatrix}
    x\\ x^q\\ x^{q^2}\\ x^{q^3}
\end{pmatrix}=M\begin{pmatrix}
    x\\ x^q\\ x^{q^2}\\ x^{q^3}
\end{pmatrix}=\begin{pmatrix}
    0\\ 0\\ 0\\ 0
\end{pmatrix},$$
and we have that $\dim_{\F_q}(\{g_m(x) : x \in \mathbb{F}_{q^5}\})=4-\text{rank}(M)$.

We want to prove that $\text{rank}(M)\geq3$ if and only if $f_\delta(Y)$ has no solutions in $\F_{q^5}$. We will indicate with $M^{c_1,c_2,\dots,c_t}_{r_1,r_2,\dots,r_s}$ the submatrix of $M$ chosen selecting the columns $\{c_1,c_2,\dots,c_t\}$ and the rows $\{r_1,r_2,\dots,r_s\}$.

Notice that $\det(M_{1,2}^{1,4})=\delta^{q+1}\neq0$. We have that $\text{rank}(M)=2$ if and only if $\det{(M_{1,2,i}^{1,4,j})}=0$ for every $i=3,4,5$ and $j=2,3$. 

Let $g_{i}^j(m):=\det (M_{1,2,i}^{1,4,j})$ and consider the system $E_{i,j,\ell}:(g_{i}^j(m))^{q^\ell}=0$ with $i=3,4,5$, $j=2,3$ and $\ell=0,1,2,3,4$. 
\begin{comment}$$\begin{cases}
    m^{q^4+q+1}\delta^{q^3+q^2+q+1}-m^{q+1} -m^q\delta^{q^3+q^2+q+2}  + m^q\delta -m^{q^4} \delta^{q^3+q^2+2q+1}-\delta^{q+1} + \delta^q=0\\
    -m^{q^3+1} + m - m^{q^3}\delta^{q+1} + m^{q^3}\delta + \delta^{q^3+q+1}=0\\
    -m^{q^3+q+1} + m^{q+1} + m^{q^3+q}\delta - m^{q^3}\delta^{q+1} + m^{q^3}\delta^q - \delta^q=0\\
    -m + m^{q^2}\delta^{q+1} - \delta^{q+1} + \delta=0\\
    m^{q^4+1}\delta^{q^3+q^2+q+1}- m - \delta^{q^3+q^2+q+2} - \delta^{q+1} + \delta=0\\
    -m^{q+1} + m^q\delta + \delta^{q^2+q+1} - \delta^{q+1} + \delta^q=0
    
\end{cases}$$
\end{comment}

In particular, $g_3^2(m)=m^q(\delta-m)+\delta^q(\delta^{q^2+1}-\delta+1)$. 
First notice that for $m=\delta $, $\text{rank}(M)\geq 3$, since  $g_3^2(\delta)= 0$  is incompatible with $\NN_{q^5/q}(\delta)=1$ by Proposition \ref{incomp}. 
\begin{itemize}
    \item Suppose that $f_\delta(m)=0$ for some $m\in \mathbb{F}_{q^5}$. Such an $m\in \mathbb{F}_{q^5}$ is not $\delta$,  since  $f_\delta(\delta)=-\delta^{q+1}(\delta^{q^2+1}-\delta+1)=0$ is incompatible with $\NN_{q^5/q}(\delta)=1$ by Proposition \ref{incomp}. Consider now $\delta\neq m$. From $g_3^2(m)=0$ we obtain $m^q=s_1(m)$, for some $s_1(Y)\in \mathbb{F}_{q^5}(Y)$. Therefore, $m^{q^2}=(s_1(m))^q$ and using $m^q=s_1(m)$ we will obtain $m^{q^2}=s_2(m)$, for some $s_2(Y)\in \mathbb{F}_{q^5}(Y)$, and analogously $m^{q^3}=s_3(m)$, $m^{q^4}=s_4(m)$, for some $s_3(Y),s_4(Y)\in \mathbb{F}_{q^5}(Y)$. Substituting these in the system we obtain that every equation is of the type $E_{i,j,\ell}:f_\delta(m)h_{i,j,\ell}(m)=0$, for some polynomial $h_{i,j,\ell}(Y)\in \mathbb{F}_{q^5}[Y]$. This shows that if $f_\delta(m)=0$ then  $\det(M_{1,2,j}^{1,4,i})=0$ for each $i=3,4,5$ and $j=2,3$. 
    \item Suppose now that $\det(M_{1,2,j}^{1,4,i})=0$ for each $i=3,4,5$ and $j=2,3$ for some $m\in \mathbb{F}_{q^5}$.  One of the polynomials $h_{i,j,\ell}(m)$ is $\delta^q$ and thus if $\det(M_{1,2,j}^{1,4,i})=0$ for each $i=3,4,5$ and $j=2,3$ then $f_\delta(m)=0$.  
\end{itemize}

To sum up, we have that the system has a solution in $\F_{q^5}$ if and only if $f_\delta(Y)$ has a root in $\F_{q^5}$. 
\end{proof}

In the following propositions, we prove that for specific choices of $\delta$, the subspace $U_\delta$ is not $\Gamma \mathrm{L}(2,q^5)$-equivalent to any of the maximum scattered subspaces from the families (C1)--(C4) listed in Theorem~\ref{Th:Classification}. In particular, Propositions \ref{Prop:C1} and \ref{Prop:C2} work for any $q$ and any $\delta$ with  $\NN_{q^5/q}(\delta)=1$.

\begin{proposition}\label{Prop:C1}
    The subspace $$U_{\delta}:=\{(x,x^{q}+\delta x^{q^4}) : \Tr_{q^5/q}(x)=0\}\leq \mathbb{F}_{q^5}^2,$$
   $\NN_{q^5/q}(\delta)=1$, is not contained up to $\Gamma \mathrm{L}(2,q^5)$-equivalence in any scattered subspace of the type 
    $$\mathrm{PR}_{s} := \{(x,x^{q^s}) : x \in \mathbb{F}_{q^5}\}, \quad s=1,\ldots,4.$$
\end{proposition}
\begin{proof}
    Since any $\mathrm{PR}_{s}$ is $\Gamma L$-equivalent to $\PR_{1}$ or $\PR_{2}$, we can divide the proof in two cases.

\noindent{\bf Case 1. $s=1$}

We have that $U_{\delta}\leq \mathrm{PR}_{1}$ if and only if there exist $A,B,C,D\in \mathbb{F}_{q^5}$, $AD\neq BC$ such that for all $x\in \mathbb{F}_{q^5}$ with $\Tr_{q^5/q}(x)=0$ there exists $y\in \mathbb{F}_{q^5}$ satisfying
$$
\begin{cases}
    Ax+B(x^q+\delta x^{q^4}) =y\\
    Cx+D(x^q+\delta x^{q^4}) =y^q.
\end{cases}$$
This yields
$$Cx+D(x^q+\delta x^{q^4}) = A^qx^q+B^q(x^{q^2}+\delta^q x)$$
for all $x\in \mathbb{F}_{q^5}$ with $\Tr_{q^5/q}(x)=0$, and so 

$$
\begin{cases}
    C-D\delta =B^q\delta^q\\
    D-D\delta=A^q\\
    -D\delta=B^q\\
    -D\delta =0.
\end{cases}$$
This implies $D=0=B=A=C$, a contradiction. 

\noindent{\bf Case 2. $s=2$}

We have that $U_{\delta}\leq \mathrm{PR}_{2}$ if and only if there exist $A,B,C,D\in \mathbb{F}_{q^5}$, $AD\neq BC$ such that for all $x\in \mathbb{F}_{q^5}$ with $\Tr_{q^5/q}(x)=0$ there exists $y\in \mathbb{F}_{q^5}$ satisfying
$$
\begin{cases}
    Ax+B(x^q+\delta x^{q^4}) =y\\
    Cx+D(x^q+\delta x^{q^4}) =y^{q^2}.
\end{cases}$$
This yields
$$Cx+D(x^q+\delta x^{q^4}) = A^{q^2}x^{q^2}+B^{q^2}(x^{q^3}+\delta^{q^2} x^q)$$
for all $x\in \mathbb{F}_{q^5}$ with $\Tr_{q^5/q}(x)=0$, and so

$$
\begin{cases}
    C-D\delta =0\\
    D-D\delta=B^{q^2}\delta^{q^2}\\
    -D\delta=A^{q^2}\\
    -D\delta =B^{q^2}.
\end{cases}$$

If $D=0$, then we obtain $A=B=C=D=0$, a contradiction. 
If $D\neq0$ we obtain $D\delta=C=-A^{q^2}=-B^{q^2}$, and so from the second equation we obtain
$1-\delta=-\delta^{q^2+1}$ that is incompatible with $\NN_{q^5/q}(\delta)=1$ by Proposition \ref{incomp}, so we have a contradiction.

\end{proof}

\begin{proposition}\label{Prop:C2}
    The subspace $$U_{\delta}:=\{(x,x^{q}+\delta x^{q^4}) : \Tr_{q^5/q}(x)=0\}\leq \mathbb{F}_{q^5}^2,$$
    with $\NN_{q^5/q}(\delta)=1$, is not contained up to $\Gamma \mathrm{L}(2,q^5)$-equivalence in any scattered subspace of the type 
    $$\mathrm{LP}_{s,\eta} := \{(x,x^{q^s}+\eta x^{q^{5-s}}) : x \in \mathbb{F}_{q^5}\}, \quad s=1,2, \quad \NN_{q^5/q}(\eta)\neq 1.$$
\end{proposition}
\begin{proof}
    Consider first the case $s=1$.
We have that $U_{\delta}\leq \mathrm{LP}_{1,\eta}$ if and only if there exist $A,B,C,D\in \mathbb{F}_{q^5}$, $AD\neq BC$ such that for all $x\in \mathbb{F}_{q^5}$ with $\Tr_{q^5/q}(x)=0$ there exists $y\in \mathbb{F}_{q^5}$ satisfying
$$
\begin{cases}
    Ax+B(x^q+\delta x^{q^4}) =y\\
    Cx+D(x^q+\delta x^{q^4}) =y^q+\eta y^{q^4}.
\end{cases}$$
This yields
\begin{equation}\label{Eq:Lp}Cx+D(x^q+\delta x^{q^4}) = A^qx^q+B^q(x^{q^2}+\delta^q x)+\eta (A^{q^4}x^{q^4}+B^{q^4}(x+\delta^{q^4} x^{q^3}))\end{equation}
for all $x\in \mathbb{F}_{q^5}$ with $\Tr_{q^5/q}(x)=0$, and so 

$$
\begin{cases}
    C-D\delta =B^q\delta^q-\eta A^{q^4}+\eta B^{q^4}\\
    D-D\delta=A^q-\eta A^{q^4}\\
    -D\delta=B^q-\eta A^{q^4}\\
    -D\delta =-\eta A^{q^4}+\eta B^{q^4}\delta^{q^4}.
\end{cases}$$
Combining the last two equations one gets $\eta B^{q^4}\delta^{q^4}=B^q$ and thus, since $$\NN_{q^5/q}(\eta \delta^{q^4})=\NN_{q^5/q}(\eta) \NN_{q^5/q}(\delta)\neq 1,$$
we obtain $B=0$. 

Thus $ D=\frac{\eta}{\delta} A^{q^4}$, and the second equation reads 
$$\Big(\frac{\eta(1-\delta)}{\delta} +\eta\Big)A^{q^4}=A^q \iff \frac{\eta}{\delta}A^{q^4}=A^q,$$
yielding again $A=0$. This is a contradiction to $AD\neq BC$.

We consider now the case $s=2$. Arguing as for $s=1$ one gets, similarly  to Equation \eqref{Eq:Lp}, 
\begin{equation}\label{Eq:Lp2}Cx+D(x^q+\delta x^{q^4}) = A^{q^2}x^{q^2}+B^{q^2}(x^{q^3}+\delta^{q^2} x^q)+\eta (A^{q^3}x^{q^3}+B^{q^3}(x^{q^4}+\delta^{q^3} x^{q^2}))\end{equation}
and thus 

$$
\begin{cases}
    C-D\delta =-\eta B^{q^3}\\
    D-D\delta=B^{q^2}\delta^{q^2}-\eta B^{q^3}\\
    -D\delta=A^{q^2}-\eta B^{q^3}+\eta B^{q^3}\delta ^{q^3}\\
    -D\delta =B^{q^2}+\eta A^{q^3}-\eta B^{q^3}.
\end{cases}$$
If $A=0$, combining the last two equations we obtain 
$$\eta \delta^{q^3}B^{q^3}=B^{q^2},$$
yielding again, by the assumptions on $\eta$ and $\delta$, $B=0$, a contradiction to $AD\neq BC$. 

Suppose that $A\neq 0$. Combining the last three equations we obtain 
$$ \delta A^{q^2} - \delta \eta A^{q^3} - \delta B^{q^2} + \delta ^{q^3+1}\eta B^{q^3}=0=(\eta -\eta \delta)A^{q^3} + (\delta^{q^2+1}-\delta+1)B^{q^2} - \eta B^{q^3}.$$

If $\delta=1$ then  $\eta B^{q^3}=B^{q^2}$, 
which is equivalent to $B=0$ due to our assumptions on $\eta$. This yields $ A^{q^2} =  \eta A^{q^3}$ and so $A=0$, again a contradiction to $AD\neq BC$.

\begin{comment}
\begin{verbatim}
clear;
K<delta,deltaq,deltaq2,deltaq3,deltaq4> := PolynomialRing(Integers(),5);
Frob := function (pol)
    return Evaluate(pol,[deltaq,deltaq2,deltaq3,deltaq4,delta]);
end function;
EQ_DELTAq3 := deltaq3*deltaq2*delta-deltaq3*delta+deltaq3-1;

DELTAq3 := -Coefficients(EQ_DELTAq3,deltaq3)[1]/Coefficients(EQ_DELTAq3,deltaq3)[2];
//IL DENOMINATORE NON SI ANNULLA ALTRIMENTI HO UNA CONTRADDIZIONE CON EQ=0

DELTAq4 := Evaluate(Frob(DELTAq3),[delta,deltaq,deltaq2,DELTAq3,deltaq4]);

EQ_DELTAq2 := Numerator(delta*deltaq*deltaq2*DELTAq3*DELTAq4-1);
//IL COEFFCIENTE DI deltaq2 NON SI ANNULLA MAI PERCHE'
//CHAR 2 -> delta !=0
//CHAR 3 -> delta=-1, IMP
//CHAR >3 IMP


DELTAq2 := -Coefficients(EQ_DELTAq2,deltaq2)[1]/Coefficients(EQ_DELTAq2,deltaq2)[2];
DELTAq3 := Evaluate(Frob(DELTAq2),[delta,deltaq,DELTAq2,deltaq3,deltaq4]);
DELTAq4 := Evaluate(Frob(DELTAq3),[delta,deltaq,DELTAq2,deltaq3,deltaq4]);
EQFIN1 := Numerator(delta*deltaq*DELTAq2*DELTAq3*DELTAq4-1);
EQFIN2 := Numerator(Evaluate(Frob(EQFIN1),[delta,deltaq,DELTAq2,deltaq3,deltaq4]));

EQFIN3 := Resultant(EQFIN1,EQFIN2,deltaq) div delta^18 div (delta-2)^21 div (delta-1) div (5*delta-1)^2;
EQFIN4 := Resultant(Frob(EQFIN3),EQFIN1,deltaq);
EQFIN5 := Resultant(Frob(EQFIN3),EQFIN3,deltaq);

Factorization(Resultant(EQFIN3,EQFIN4,delta));

\end{verbatim}
\end{comment}
From now on we can suppose  that $\delta\neq1$. Thus
$$A^{q^2}= \frac{\delta^{q^2+1} B^{q^2} +\eta(- \delta^{q^3+1}+ \delta^{q^3} - 1)B^{q^3}}{\delta-1}, \qquad A^{q^3} = \frac{\delta^{q^3+q} B^{q^3} +\eta^q(- \delta^{q^4+1}+ \delta^{q^4} - 1)B^{q^4}}{\delta^q-1}.$$

From $\delta A^{q^2} - \delta \eta A^{q^3} - \delta B^{q^2} + \delta ^{q^3+1}\eta B^{q^3}=0$, 
$$(\delta - 1)^q\delta (\delta^{q^2+1} - \delta + 1)B^{q^2}+\eta \delta(\delta^{q^3+q+1} - \delta^{q^3+q}+ \delta^q- 1)B^{q^3}+\eta^{q+1}(\delta^{q^4+q} - \delta^{q^4} + 1)\delta (\delta-1)B^{q^4}.$$
The determinant of the Dickson matrix  of the above linearized polynomial, when considering also $\NN_{q^5/q}(\delta)=1$ factorizes as 

\begin{eqnarray*}
    (\NN_{q^5/q}(\eta)-1)^2\delta^{q^3+2q^2 +q+1}(\delta-1)^{q^3+q^2+q+1}(\delta^{q^2+1} - \delta + 1)^{q^3+q+1}\cdot\\(\delta^{q^3+q^2+q+1}+\delta^q-1)
(\delta^{q^3+q^2+q+1}-1)(\delta^{q^3+q^2+q+1}-\delta^{q^3+q+1}-1)=0,
\end{eqnarray*}
and we have a non-zero solution $B$ if and only if one of these factors is vanishing.

The first three are trivially non-zero. 

\begin{enumerate}
\item $\delta^{q^3+q^2+q+1}=1$. This would yield $\delta=1$, a contradiction.
\item $\delta^{q^2+1} - \delta + 1=0$. This would yield  $\delta^{q^4}=\frac{1}{1-\delta}$, $\delta^{q}=\delta\in \mathbb{F}_q$, and thus $\delta^5=1$. Now,  $\delta^5=1$ and $\delta^2 - \delta + 1=0$ provide a contradiction since $\gcd(x^5-1,x^2-x+1)=1$. 
\item $\delta^{q^3+q^2+q+1}+\delta^q-1=0$. This would yield $\delta^{q^4+q^3+q^2+q}=1-\delta^{q^2}$ and thus $1=\delta-\delta^{q^2+1}$, that is $\delta^{q^2}=\frac{\delta-1}{\delta}$, a contradiction as above.
\item $\delta^{q^3+q^2+q+1}-\delta^{q^3+q+1}-1=0$. This would yield  $\delta^{q^4+q^3+q^2+q}-\delta^{q^4+q^2+q}-1=0$ and then, multiplying by $\delta$, $1-\delta^{q^4+q^2+q+1}-\delta=0$. Raising it to the power $q$, one obtains $$1-\delta^{q^3+q^2+q+1}-\delta^q=0$$ and a contradiction as above.
    
\end{enumerate}



 Thus the only possibility is $B=0$, a contradiction since from $\delta A^{q^2} - \delta \eta A^{q^3} - \delta B^{q^2} + \delta ^{q^3+1}\eta B^{q^3}=0$ and our assumptions on $\eta$ one gets $A=0$.
\end{proof}

Note that when $q=3^{2h+1}$, by Proposition \ref{Prop:Cond}, $U_1$ is scattered and of dimension $4$,  since $5$ is not a square in $\mathbb{F}_q$.  In the following we prove that, with this particular choice of $q$ and with $\delta=1$, $U_1$ is not contained up to $\Gamma \mathrm{L}(2,q^5)$-equivalence in any subspace $\WW_{\eta,\rho}$.

\begin{proposition}\label{Prop:C3}
    Let $q=3^{2h+1}$. The subspace
    $$U_{1}:=\{(x,x^{q}+ x^{q^4}) : \Tr_{q^5/q}(x)=0\}\leq \mathbb{F}_{q^5}^2,$$
    is not contained in any of the subspaces in the $\Gamma \mathrm{L}(2,q^5)$-orbit of 
    $$\WW_{\eta,\rho}:= \{(\eta (x^q - x) + \Tr_{q^5/q}(\rho x), x^q - x^{q^4}) : x\in \mathbb{F}_{q^5}\},$$
    for every $\eta,\rho\in\F_{q^5}$ with $\Tr_{q^5/q}(\eta)=0\neq \Tr_{q^5/q}(\rho)$.
\end{proposition}
\begin{proof}
    The subspace $U_1$ is contained in any subspace within the $\mathrm{GL}$-orbit of $\WW_{\eta,\rho}$ if and only if there exist $A,B,C,D\in\F_{q^5}$ with $AD\neq BC$, such that for every $x\in \mathbb{F}_{q^5}$ with $\Tr_{q^5/q}(x)=0$, there exists a $y\in\F_{q^5}$ for which the following holds:
$$
\begin{pmatrix}
A&B\\C&D
\end{pmatrix}\begin{pmatrix}
x\\x^q+x^{q^4}
\end{pmatrix}=\begin{pmatrix}
\eta (y^q - y) + \Tr_{q^5/q}(\rho y)\\y^q - y^{q^4}
\end{pmatrix}.
$$
This matrix equation yields two separate equations:
\begin{eqnarray}
\label{eq1}Ax+B(x^q+x^{q^4})&=&\eta (y^q - y) + \Tr_{q^5/q}(\rho y),\\
\label{eq2}Cx+D(x^q+x^{q^4})&=&y^q - y^{q^4}.
\end{eqnarray}
Let us first focus on Equation \eqref{eq2}, which can be rewritten as $Cx+C(x^q+x^{q^4})-y^q + y^{q^4}=0$. We can then consider the system of equations $(Cx+C(x^q+x^{q^4})-y^q + y^{q^4})^{q^\ell}=0$ for $\ell=0,\dots,4$. From the equations corresponding to $\ell=1$ and $\ell=4$, we obtain expressions for $y^{q^2}$ and $y^{q^3}$, respectively. These can be written as $y^{q^2}=f_2(y)$ and $y^{q^3}=f_3(y)$, for some polynomials $f_2(Y),f_3(Y)\in \mathbb{F}_{q^5}[Y]$. We can then substitute these expressions into the other equations. Next, from the equation with $\ell=3$, we obtain $y^{q^4}=f_4(y)$, where $f_4(Y)\in \mathbb{F}_{q^5}[Y]$. Once these substitutions are made, the equations for $\ell=0$ and $\ell=2$ become
    \begin{equation*}\begin{cases}
        2Cx + 2C^qx^q + 2C^{q^3}x^{q^3} + 2D(x^q + x^{q^4}) + 2D^q(x + x^{q^2}) + 2D^{q^3}(x^{q^2} + x^{q^4}) + 2y + y^q=0,\\
    2C^{q^2}x^{q^2} + 2C^{q^4}x^{q^4} + 2D^{q^2}x^q + 2D^{q^2}x^{q^3} + 2D^{q^4}x + 2D^{q^4}x^{q^3} + y + 2y^{q}=0.
    \end{cases}
    \end{equation*}
From the second equation, we can obtain $y^q = f_1(y)$, where $f_1(Y) \in \mathbb{F}_{q^5}[Y]$. By substituting this into the first equation, we get the following:
\begin{eqnarray*}
(2C^{q^4} + 2D + 2D^{q^3})x^{q^4} &+& (2C^{q^3} + 2D^{q^2} + 2D^{q^4})x^{q^3} + (2C^{q^2} + 2D^{q} +2D^{q^3})x^{q^2}\\
&+& (2C^q+2D+2D^{q^2})x^q + (2C+ 2D^{q} +2D^{q^4})x=0.
\end{eqnarray*}
By substituting $x^{q^4}=-x^{q^3}-x^{q^2}-x^q-x$, we obtain a polynomial in $x$ of degree $q^3$. Since this polynomial has at least $q^4$ solutions, it must be the zero polynomial, which means all of its coefficients must be equal to zero. This leads to the following system of equations:
\begin{equation}\begin{cases}\label{CoeffEq2}
C - C^{q^4} - D + D^q - D^{q^3} + D^{q^4}=0\\
C^q - C^{q^4} + D^{q^2} - D^{q^3}=0\\
C^{q^2} - C^{q^4} - D + D^q=0\\
C^{q^3} - C^{q^4} - D + D^{q^2} - D^{q^3} + D^{q^4}=0.
\end{cases}
\end{equation}

    We can apply a similar process to Equation \eqref{eq1}. Consider the expression $Ax+B(x^q+x^{q^4})-\eta (y^q - y) - \Tr_{q^5/q}(\rho y)=0$  and so the system of equations $(Ax+B(x^q+x^{q^4})-\eta (y^q - y) - \Tr_{q^5/q}(\rho y))^{q^\ell}=0$ for $\ell=0,\dots,4$. By substituting the expressions for $y^{q^i}=f_i(y)$ with $i=1,\dots,4$, the third equation becomes:
\begin{eqnarray*}
&&(A^{q^4} + C^{q^4}\rho^q + C^{q^4}\rho^{q^3} + D^{q^3}\eta^{q^4} + 2D^{q^3}\rho^{q^4})x^{q^4}\\ &+& (B^{q^4} + C^{q^3}\eta^{q^4} + 2C^{q^3}\rho^{q^4} + D^{q^2}\rho^q + D^{q^4}\rho^q + D^{q^4}\rho^{q^3})x^{q^3}\\
&+& (C^{q^2}\rho^q + D^q\eta^{q^4} + 2D^q\rho^{q^2} + 2D^q\rho^{q^4} + D^{q^3}\eta^{q^4} + 2D^{q^3}\rho^{q^4})x^{q^2}\\
&+& (C^q\eta^{q^4} + 2C^q\rho^{q^2}  +2C^q\rho^{q^4} + D^{q^2}\rho^q)x^q\\
&+& (B^{q^4} + D^q\eta^{q^4} + 2D^q\rho^{q^2} + 2D^q\rho^{q^4} + D^{q^4}\rho^q + D^{q^4}\rho^{q^3})x+y(\Tr_{q^5/q}(\rho))=0.
\end{eqnarray*}
Since $\Tr_{q^5/q}(\rho) \neq 0$, we can express $y$ as a function of $x$, i.e., $y=f(x)$. We can then substitute this expression for $y$ into the other equations. This will give us four equations that depend solely on $x$. As before, by substituting $x^{q^4}=-x^{q^3}-x^{q^2}-x^q-x$, we will obtain four polynomials in $x$. Since these polynomials must be the zero polynomial, every coefficient of each polynomial must be zero, leading to the following system of equations:
\begin{equation}\begin{cases}\label{CoeffEq1}
2A^{q^4} + B^{q^3} + B^{q^4} + 2C^{q^4}\eta^{q^3} + D^q\eta^{q^3} + D^q\eta^{q^4} + 2D^{q^3}\eta^{q^3} + 2D^{q^3}\eta^{q^4} + D^{q^4}\eta^{q^3}=0\\
2A^{q^4} + 2B^q + C^{q^2}\eta^q + 2C^{q^4}\eta^q + D^q\eta^q + D^q\eta^{q^4}=0\\
2A^{q^4} + B^{q^3} + C^q\eta^{q^3} + C^q\eta^{q^4} + 2C^{q^4}\eta^{q^3} + 2D^{q^3}\eta^{q^3} + 2D^{q^3}\eta^{q^4}=0\\
2A + 2A^{q^4} + B + B^{q^4} + C^{q^4}\eta + D^q\eta^{q^4} + 2D^{q^3}\eta^{q^4} + 2D^{q^4}\eta=0\\
2A^{q^2} + 2A^{q^4} + C^{q^4}\eta^{q^2} + 2D^q\eta^{q^2} + D^q\eta^{q^4}=0\\
2A^q + 2A^{q^4} + C^q\eta^q + C^q\eta^{q^4} + 2C^{q^4}\eta^q + D^{q^2}\eta^q + 2D^{q^3}\eta^{q^4}=0\\
2A^{q^4} + C^q\eta^{q^4} + C^{q^4}\eta + 2D^{q^2}\eta + 2D^{q^3}\eta^{q^4}=0\\
2A^{q^4} + 2B^q + B^{q^4} + 2C^{q^4}\eta^q + D^q\eta^q + D^q\eta^{q^4} + 2D^{q^3}\eta^{q^4} + D^{q^4}\eta^q=0\\
2A^{q^4} + 2B^{q^2} + 2C^q\eta^{q^2} + C^q\eta^{q^4} + C^{q^4}\eta^{q^2} + 2D^{q^3}\eta^{q^4}=0\\
2A^{q^4} + B^{q^4} + C^{q^3}\eta^{q^4} + 2C^{q^4}\eta^q + D^{q^2}\eta^q + 2D^{q^3}\eta^{q^4} + D^{q^4}\eta^q=0\\
2A^{q^4} + B^{q^4} + C^{q^4}\eta^{q^2} + 2D^q\eta^{q^2} + D^q\eta^{q^4} + 2D^{q^3}\eta^{q^4} + 2D^{q^4}\eta^{q^2}=0\\
2A^{q^4} + B + B^{q^4} + C^{q^3}\eta^{q^4} + C^{q^4}\eta + 2D^{q^2}\eta + 2D^{q^3}\eta^{q^4} + 2D^{q^4}\eta=0\\
2A^{q^4} + B + 2C^{q^2}\eta + C^{q^4}\eta + D^q\eta^{q^4}=0\\
2A^{q^4} + 2C^{q^4}\eta^{q^3} + D^q\eta^{q^3} + D^q\eta^{q^4}=0\\
2A^{q^4} + 2B^{q^2} + B^{q^4} + C^{q^3}\eta^{q^4} + C^{q^4}\eta^{q^2} + 2D^{q^3}\eta^{q^4} + 2D^{q^4}\eta^{q^2}=0\\
2A^{q^3} + 2A^{q^4} + B^{q^3} + B^{q^4} + C^{q^3}\eta^{q^3} + C^{q^3}\eta^{q^4} + 2C^{q^4}\eta^{q^3} + 2D^{q^3}\eta^{q^3} + 2D^{q^3}\eta^{q^4} + D^{q^4}\eta^{q^3}=0.
\end{cases}
\end{equation}

    From the equations in \eqref{CoeffEq2}, we obtain expressions for $C^{q^i}$ in terms of $C^{q^4}$ for $i=0,1,2,3$. We can substitute these into \eqref{CoeffEq1}.

Let $\eta=0$. Then $A=0$ from  the third-to-last equation in \eqref{CoeffEq1}. 
From the fourth-to-last equation, we obtain $B=0$. This contradicts the assumption that $AD \neq BC$. Therefore, from now on we make the following assumption
\begin{enumerate}
    \item[(A1)]  $\eta\neq 0$.
\end{enumerate}
Again from the third-to-last equation in \eqref{CoeffEq1},
we can express $C^{q^4}$ in terms of $D^q$, and then substitute this into the other equations.

The second equation in \eqref{CoeffEq1} becomes:
\begin{equation}\label{findA}
A^{q^4} + B^q + 2D\eta^q + 2D^q\eta^{q^4}=0,
\end{equation}
from which we find $A=2B^{q^2}+D^{q}\eta^{q^2} + D^{q^2}\eta$. We can substitute this into the other equations.

The fourth-to-last equation in \eqref{CoeffEq1} becomes $B + B^q + 2D\eta + 2D\eta^q + D^q\eta = 0$. Similarly, we can find expressions for $B^{q^i}$ in terms of $B^{q^4}$ for $i=0,1,2,3$, and substitute these. The eighth equation in \eqref{CoeffEq1} becomes:
\begin{equation}\label{findB}
B^{q^4}(\eta^q+2\eta^{q^3}) + D(\eta^{q+1}+\eta^{q^3+q}+\eta^{q^4+q})+2D^q\eta^{q+1} D^{q^3}\eta^{q^4+q^3}+D^{q^4}(2\eta^{q+1}+2\eta^{q^3+q}+2\eta^{q^4+q})=0.
\end{equation}
Noting that $\eta^q+2\eta^{q^3}\neq0$ since $\eta\neq0$, we can express $B^{q^4}$ in terms of the variables represented by $D$, and substitute this into the other equations. After removing the dependent equations, we are left with
    \begin{equation*}
        \begin{cases}\label{SystemofaD}
            D(\eta+\eta^{q^3}+\eta^{q^4})(\eta+\eta^{q}+\eta^{q^4})+D^q\eta(\eta+\eta^{q}+\eta^{q^4})+D^{q^2}(\eta+\eta^{q^4})(\eta^q+2\eta^{q^3})+\\
            \hspace{0.5 cm}+D^{q^3}\eta^{q^4}(\eta+\eta^{q^3}+\eta^{q^4})+D^{q^4}(\eta+\eta^{q^3}+\eta^{q^4})(\eta+\eta^{q}+\eta^{q^4})&=0,\\
            D(\eta^{q^3+1} + \eta^{q^4+1} + 2\eta^{2q} + 2\eta^{q+q^3} + 2\eta^{q^4+q^3} + \eta^{2q^4})\\
            \hspace{0.5 cm}+D^q(\eta^{q^3+1} + \eta^{q^4+1} + 2\eta^{2q} + \eta^{q^3+q})+D^{q^2}(\eta^q+2\eta^{q^4})(\eta^q+2\eta^{q^3})\\
            \hspace{0.5 cm}+D^{q^3}\eta^{q^4}(\eta^q + \eta^{q^3} + 2\eta^{q^4})+D^{q^4}(\eta^{q^3+1} + \eta^{q^4+1} + \eta^{q+q^3} + 2\eta^{q^4+q^3} + \eta^{2q^4})&=0,\\
            D(\eta^{q+1} + \eta^{q^2+1} + \eta^{q^3+1} + \eta^{q^2+q} + 2\eta^{q^3+q} + \eta^{q^4+q} + \eta^{q^4+q^2} + \eta^{q^4+q^3})\\
            \hspace{0.5 cm}+D^q\eta(\eta^q+\eta^{q^2}+\eta^{q^3})+D^{q^3}\eta^{q^4}(\eta^q+\eta^{q^2}+\eta^{q^3})\\
            \hspace{0.5 cm}+D^{q^4}(\eta^{q+1} + \eta^{q^2+1} + \eta^{q^3+1} +  2\eta^{q^3+q} + \eta^{q^4+q} + 
            \eta^{q^3+q^2}+ \eta^{q^4+q^2} + \eta^{q^4+q^3})&=0,\\
            D(\eta^{q^3+1} + 2\eta^{q^4+1} + \eta^{2q} + 2\eta^{q^4+q} + \eta^{q^4+q^3} + 2\eta^{2q^4})\\
            \hspace{0.5 cm}+D^q(\eta^{q^3+1} + 2\eta^{q^4+1} + \eta^{2q} + \eta^{q^2+q} + 2\eta^{q^3+q} + 2\eta^{q^3+q^2})\\
            \hspace{0.5 cm}+D^{q^2}(\eta^q+2\eta^{q^3})(\eta^q+\eta^{q^2}+2\eta^{q^4})+D^{q^3}(\eta^{q^2+q} + 2\eta^{q^3+q^2} + 2\eta^{q^4+q^3} + \eta^{2q^4})\\
            \hspace{0.5 cm}+D^{q^4}(\eta^{q^3}+2\eta^{q^4})(\eta+\eta^q+\eta^{q^4})&=0,\\
            D(\eta^{q+1} + 2\eta^{q^3+1} + 2\eta^{q^4+1} + 2\eta^{q^4+q^3} + 2\eta^{2q^4})+D^q\eta(\eta^q+2\eta^{q^3}+2\eta^{q^4})\\ \hspace{0.5 cm}+D^{q^2}(\eta^{q^3}+\eta^{q^4})(\eta^{q}+2\eta^{q^3})++D^{q^3}(\eta^{q^3+q} + \eta^{q^4+q} + 2\eta^{2q^3} + 2\eta^{q^4+q^3} + 2\eta^{2q^4})+\\
            \hspace{0.5 cm}+D^{q^4}(\eta^{q+1} + 2\eta^{q^3+1} + 2\eta^{q^4+1} + \eta^{q^3+q}+ \eta^{2q^3}+ 2\eta^{q^4+q^3} + 2\eta^{2q^4})&=0.
        \end{cases}
    \end{equation*}
    Notice that if $D=0$, then from \eqref{findB} we obtain $B=0$, which contradicts the assumption that $AD \neq BC$. Thus, we can consider $D \neq 0$. This would yield a solution $(D,D^q,D^{q^2},D^{q^3},D^{q^4})\neq\mathbf{0}$ to the system of equations in (\ref{SystemofaD}). If we consider \eqref{SystemofaD} to be a linear system with unknowns $(D,D^q,D^{q^2},D^{q^3},D^{q^4})$, the associated matrix must be singular. Therefore, we can require the determinant to be zero. Let $\omega$ be the generator of $\mathbb{F}_9^*$, where $w^2 + 2w + 2=0$. We have that the determinant is equal to $2c_1c_2c_3$, where
    \begin{eqnarray*}
   c_1&=&\eta^q + 2\eta^{q^3},\\
   c_2&=&\eta^{q^2+q+1} + \omega^2\eta^{q^3+q+1} + \eta^{q^4+q+1} + \omega^2\eta^{q^3+q^2+1} + \omega^2\eta^{q^4+q^2+1} + \eta^{q^4+q^3+1} + \eta^{q^3+q^2+q}+\\&& +
        \omega^2\eta^{q^4+q^2+q} + \omega^2\eta^{q^4+q^3+q} + \eta^{q^4+q^3+q^2},\\
        c_3&=&\eta^{q^2+q+1} + \omega^6\eta^{q^3+q+1} + \eta^{q^4+q+1} + \omega^6\eta^{q^3+q^2+1} + \omega^6\eta^{q^4+q^2+1} + \eta^{q^4+q^3+1} + \eta^{q^3+q^2+q}+\\&& +
        \omega^6\eta^{q^4+q^2+q} + \omega^6\eta^{q^4+q^3+q} + \eta^{q^4+q^3+q^2}.
    \end{eqnarray*}
    We first note that $c_1 \neq 0$. We can also see that $c_2^{q^5}=c_3$, which allows us to consider a case where both coefficients are zero. In this case, we can substitute the expression for $\eta^{q^4}$ given by $\eta^{q^4}=2\eta^{q^3}+2\eta^{q^2}+2\eta^{q}+2\eta$ into both equations and then proceed to consider  both the following expressions to be vanishing 
    \begin{eqnarray*}
        t_1&:=&\frac{c_2+c_3}{2}=\eta^{q+1}(\eta+\eta^q)+\eta^{q^3}(\eta^2 + 2\eta^{q+1} + 2\eta^{q^2+1} + \eta^{2q^2})+\eta^{2q^3}(\eta+\eta^{q^2})\in\mathbb{F}_{q^5},\\
        t_2&:=&\frac{c_2-c_3}{\omega^6}=\eta^{q^2}(\eta+\eta^q)(\eta+\eta^q+\eta^{q^2})+\eta^{q^3}(\eta^q)(\eta^q+2\eta^{q^2})+\eta^{2q^3}\eta^q\in\mathbb{F}_{q^5}.
    \end{eqnarray*}

       Let $t_3:=(\eta+\eta^{q^2})t_2-\eta^qt_1=2(\eta+\eta^q)(\eta^{q^2+2} + 2\eta^{2q+1} + \eta^{q^2+q+1} + 2\eta^{2q^2+1} + \eta^{2q^2+q} + \eta^{3q^2})+\eta^{q^3}\eta^q(\eta+2\eta^{q^2})(\eta+\eta^q+\eta^{q^2})$. Since $\eta^q(\eta+2\eta^{q^2})(\eta+2\eta^{q^2})\neq0$ from the assumption on $\eta$, we can express $\eta^{q^3}$ as a function of $\eta, \eta^q, \text{ and } \eta^{q^2}$, i.e., $\eta^{q^3}=h(\eta,\eta^q,\eta^{q^2})$. We can then substitute this into the expression for $t_2$ to get:
$$\eta^{q^2+q+1}(\eta+\eta^{q})^{q+1}(\eta+\omega\eta^q+\eta^{q^2})^2(\eta+\omega^3\eta^q+\eta^{q^2})^2=0.$$
The first two factors are clearly non-zero. As before, we can observe that the third and fourth factors are expressions in $\mathbb{F}_{q^{10}}$. By setting both to zero and then combining them, we arrive at the conclusion that $\eta=0$, which is a contradiction to our  assumption (A1).
\end{proof}

In order to deal with the case (C4), we first obtain necessary conditions on $k\in \mathbb{F}_{q^5}$, $q=3^{2h+1}$, for $U_1$ to be contained in $\ZZ_{k}$.

\begin{proposition}\label{Prop:Zk_1}
    Let $q=3^{2h+1}$. If $U_1$ is contained, up to $\Gamma \mathrm{L}(2,q^5)$ equivalence, in $\ZZ_k$, $\NN_{q^5/q}(k)=1$, then  $u_1(k)u_2(k)=0$ and $t(k)=0$, where 
    \begin{eqnarray*}
    u_1(k) &:=& k^{2+2q+2q^{2}}+2 k^{1+2q+2q^{2}}+2 k^{1+2q+q^{2}}+ k^{1+q+2q^{2}}+2 k^{1+q^{2}}+ k^{q+q^{2}}+2 k^{q^{2}}+1;\\
    u_2(k)&:=& k^{2+2q+2q^{2}}+2 k^{2+2q+q^{2}}+ k^{2+q+q^{2}}+2 k^{1+2q+q^{2}}+ k^{1+q}+2 k^{1+q^{2}}+2 k+1;\\
    t(k)&:=&k^{1+q+q^2}(k^{3+3q+3q^{2}+q^{3}}+2 k^{3+3q+2q^{2}+q^{3}}+2 k^{2+3q+3q^{2}+q^{3}}+k^{2+2q+q^{2}+q^{3}}+k^{2+2q+q^{2}}\\&&+2 k^{2+q+2q^{2}+q^{3}}+k^{2+q+q^{2}+q^{3}}
    +2 k^{2+q+q^{2}}+k^{1+3q+2q^{2}+q^{3}}+k^{1+2q+3q^{2}+q^{3}}+2 k^{1+2q+2q^{2}}\\
    &&+2 k^{1+2q+q^{2}+q^{3}}+2 k^{1+q+3q^{2}+q^{3}}+k^{1+q+2q^{2}}+2 k^{1+q}+k^{1+2q^{2}+q^{3}}+2 k^{1+q^{3}}+k\\
    &&+2 k^{2q+2q^{2}+q^{3}}+k^{q+2q^{2}+q^{3}}+k^{q+q^{2}}+2 k^{q^{2}+q^{3}}+2 k^{q^{2}}+k^{q^{3}}).
    \end{eqnarray*}
\end{proposition}
\begin{proof}
    Let us consider  $k\neq 1$.

    The subspace $U_1$ is contained in any of the subspaces in the $\Gamma \mathrm{L}$-orbit of $\ZZ_k$ if and only if there exist $A,B,C,D\in\F_{q^5}$ with $AD\ne BC$, such that for every $x\in \mathbb{F}_{q^5}$, $\Tr_{q^5/q}(x)=0$, there exists $y\in\F_{q^5}$ for which 

    $$
    \begin{cases}
        y=A x +B (x^q+x^{q^4}),\\
k(y^q+y^{q^3})+y^{q^2}+y^{q^4}=C x +D (x^q+x^{q^4}).
    \end{cases}$$
    
    From the system above we obtain 
    $$\begin{cases}
         A^{q^4} + 2 B^{q} k + B^{q^3} k + 2 B^{q^4} + C + 2D&=0,\\
    2 A^{q} k + A^{q^4} + 2 B^{q^2} + B^{q^3} k&=0,\\
    2 A^{q^2} + A^{q^4} + 2 B^{q} k + 2 D&=0,\\
    2 A^{q^3} k + A^{q^4} + 2 B^{q^2} + B^{q^3} k + 2 B^{q^4} + 2D&=0,
    \end{cases}$$
and thus 
$$A^{q^2} + 2 A^{q^3} k + B^{q} k + 2 B^{q^2} + B^{q^3} k + 2 B^{q^4}=0=2 A^{q} k + A^{q^4} + 2 B^{q^2} + B^{q^3}k. $$
In particular, considering the $q$-Frobenius of the above quantities and combining them,
$$\begin{cases}
         2 A^{q} k + A^{q^4} + 2 B^{q^2} + B^{q^3} k&=0,\\
    2 A k^{q^4} + A^{q^3} + 2 B^{q} + B^{q^2} k^{q^4}&=0,\\
    A k^{q^4+1} + 2 A^{q^2} + 2 B^{q^2} k^{q^4+1} + B^{q^2} + 2 B^{q^3} k + B^{q^4}&=0,\\
    2 A k^{2q^4+1} + A^{q} + B k^{q^4} + 2 B^{q} + B^{q^2} k^{2q^4+1} + B^{q^3} k^{q^4+1} + 2 B^{q^3} + 2 B^{q^4} k^{q^4}&=0,\\
    A k^{q^4+q+1} + 2 A + 2 B^{q^2} k^{q^4+q+1} + B^{q^2} k^{q} + 2 B^{q^3} k^{q+1} + B^{q^3}&=0,\\
    2 B k^{2q^4+q^3+q+2} + B k^{q^4+q+1} + B k^{q^4+q^3+1} + 2 B + B^{q}k^{q^4+q^3+q+2}\\
    \hspace{0.5 cm}+ 2 B^{q}k^{q^4+q^3+q+1}+ 
    2 B^{q}k^{q^3+1} + B^{q} k^{q^3} + 2 B^{q^2} k^{2q^4+q^3+q+2} + B^{q^2} k^{2q^4+q^3+2}  \\
    \hspace{0.5 cm}+ B^{q^2}k^{q^4+q^3+q+1}+ 2 B^{q^2}k^{q^4+1} + 
    2 B^{q^2} k^{q^3} + B^{q^2} + 2 B^{q^3} k^{2q^4+q^3+2} + B^{q^3} k^{q^4+q^3+2}  \\ 
    \hspace{0.5 cm}+ B^{q^3}k^{q^4+1} + 2 B^{q^3} k+ B^{q^4}k^{2q^4+q^3+q+2}+ 
    2 B^{q^4}k^{q^4+q+1} + 2 B^{q^4}k^{q^4+q^3+1} + B^{q^4}&=0.
    \end{cases}$$

Also, the determinant of the Dickson matrix with respect to $B$ of the last linearized polynomial in the above system factorizes as 
    \begin{eqnarray*}
        \Big(\NN_{q^5/q}(k^{1+q+q^2}-1)\Big)\Big(\NN_{q^5/q}(k)+1\Big)\Big(k^{1+q+q^{2}+q^{3}+q^{4}}+\xi k^{1+q+q^{2}+q^{3}}+\xi k^{1+q+q^{2}+q^{4}}+2 k^{1+q+q^{2}}\\+\xi k^{1+q+q^{3}+q^{4}}+2 k^{1+q+q^{4}}+k^{1+q}+\xi k^{1+q^{2}+q^{3}+q^{4}}+2 k^{1+q^{3}+q^{4}}+k^{1+q^{4}}+\xi^5 k+\xi k^{q+q^{2}+q^{3}+q^{4}}\\+2 k^{q+q^{2}+q^{3}}+k^{q+q^{2}}+\xi^5 k^{q}+2 k^{q^{2}+q^{3}+q^{4}}+k^{q^{2}+q^{3}}+\xi^5 k^{q^{2}}+ k^{q^{3}+q^{4}}+\xi^5 k^{q^{3}}+\xi^5 k^{q^{4}}+2)^{1+q^5}\Big),
    \end{eqnarray*}
where $\xi \in \mathbb{F}_{9}$ satisfies $\xi^2 + 2\xi + 2=0$. 

Suppose that the above quantity is different from zero. Then the unique solution is $B=0$ and from 
$$A k^{q^4+q+1} + 2 A + 2 B^{q^2} k^{q^4+q+1} + B^{q^2} k^{q} + 2 B^{q^3} k^{q+1} + B^{q^3}=0$$
we obtain $A=0$ since $k^{q^4+q+1}=1$ and $\NN_{q^5/q}(k)=1$ would yield $k=1$, a contradiction. 

This means that the determinant of the Dickson matrix must vanish. 

The first 2 factors are different from $0$ since, together with $\NN_{q^5/q}(k)=1$, they yield  a contradiction to  our assumptions on $k$. Since  $\mathbb{F}_{9}\not \leq \mathbb{F}_{q^5}$, if the last factor, say $g(k)$ for some $g(Y)\in\F_{q^5}[Y]$, is zero then
$$g(k)+(g(k))^{q^5}=0=g(k)-(g(k))^{q^5},$$
i.e., 
$$
\begin{cases}
    2 k^{1+q+q^{2}+q^{3}+q^{4}}+ k^{1+q+q^{2}+q^{3}}+ k^{1+q+q^{2}+q^{4}}+ k^{1+q+q^{2}}+ k^{1+q+q^{3}+q^{4}}+ k^{1+q+q^{4}}+2 k^{1+q}\\
    \hspace{1 cm}+ k^{1+q^{2}+q^{3}+q^{4}}+ k^{1+q^{3}+q^{4}}+2 k^{1+q^{4}}+2 k+ k^{q+q^{2}+q^{3}+q^{4}}+ k^{q+q^{2}+q^{3}}+2 k^{q+q^{2}}+2 k^{q}\\
    \hspace{1 cm}+ k^{q^{2}+q^{3}+q^{4}}+2 k^{q^{2}+q^{3}}+2 k^{q^{2}}+2 k^{q^{3}+q^{4}}+2 k^{q^{3}}+2 k^{q^{4}}+1=0,\\
k^{1+q+q^{2}+q^{3}}+ k^{1+q+q^{2}+q^{4}}+ k^{1+q+q^{3}+q^{4}}+ k^{1+q^{2}+q^{3}+q^{4}}+2 k+ k^{q+q^{2}+q^{3}+q^{4}}\\
\hspace{1 cm}+2 k^{q}+2 k^{q^{2}}+2 k^{q^{3}}+2 k^{q^{4}}=0.
\end{cases}$$
Combining them with $\NN_{q^5/q}(k)=1$ we obtain 
$$
\begin{cases}
  s_1(k):=2 k^{2+2q+2q^{2}+2q^{3}}+2 k^{2+2q+2q^{2}+q^{3}}+k^{2+2q+q^{2}+q^{3}}+k^{2+q+q^{2}+q^{3}}+2 k^{1+2q+2q^{2}+2q^{3}}\\
  \hspace{1 cm}+k^{1+2q+2q^{2}+q^{3}}+k^{1+2q+q^{2}+q^{3}}+k^{1+q+2q^{2}+2q^{3}}  +k^{1+q+2q^{2}+q^{3}}+k^{1+q+q^{2}+2q^{3}}+2 k^{1+q+q^{2}}\\
  \hspace{1 cm}+2 k^{1+q+q^{3}}+2 k^{1+q}+2 k^{1+q^{2}+q^{3}}+2 k^{1+q^{3}}+k+2 k^{q+q^{2}+q^{3}}+2 k^{q^{2}+q^{3}}+k^{q^{3}}+1=0\\
        s_2(k):=2 k^{2+2q+2q^{2}+2q^{3}}+k^{2+q+q^{2}+q^{3}}+k^{1+2q+q^{2}+q^{3}}+k^{1+q+2q^{2}+q^{3}}+k^{1+q+q^{2}+2q^{3}}+2 k^{1+q+q^{2}}\\
        \hspace{1 cm}+2 k^{1+q+q^{3}}+2 k^{1+q^{2}+q^{3}}+2 k^{q+q^{2}+q^{3}}+1=0,\\
\end{cases}
$$
and thus eliminating $k^{2q^3}$ from the two equations above, we obtain $t(k)=0$.

Getting $k^{q^3}$ from $t(k)=0$ and substituting it in $s_2(k)$ we obtain 
$$(k-1)^{1+q^2} (k^{q+1}-1)^{q+1}(k^{q^2+q+1}-1)u_1(k)u_2(k)=0,$$
and the claim follows, since $(k-1)^{1+q^2} (k^{q+1}-1)^{q+1}(k^{q^2+q+1}-1)=0$ implies $k=1$. 
\end{proof}

The following proposition deals with the case (C4). 
\begin{proposition}\label{Prop:C4}
    Let $q=3^{2h+1}$, $h>1$. If $U_1$ is scattered and of dimension $4$ and it is contained in $$\ZZ_{k}:= \{(x,k(x^{q}+x^{q^3})+x^{q^2}+x^{q^4}) : x\in \mathbb{F}_{q^5}\},$$
with $\NN_{q^5/q}(k)=1$,  then $\ZZ_{k}$ is not scattered.
\end{proposition}
\begin{proof}
By Proposition~\ref{Prop:Zk_1}, we can assume that $u_1(k)u_2(k)=t(k)=0$. Note that for $k=1$, the set $\ZZ_k$ is not scattered. In what follows, we consider $k \neq 1$. We will prove that for all $k \neq 1$ satisfying $u_1(k)u_2(k)=t(k)=0$, the set $\ZZ_k$ is not scattered. To this end, we will show that there exists an element $m \in \F_{q^5}$ such that
\begin{equation}\label{rankZk}
    \rank(mx+k(x^q+x^{q^3})+x^{q^2}+x^{q^4})=3.
\end{equation}
Consider the Dickson matrix $M$ associated with the linearized polynomial in \eqref{rankZk}:
\begin{equation*}
    M=\begin{pmatrix}
        m&k&1&k&1\\
        1&m^q&k^q&1&k^q\\
        k^{q^2}&1&m^{q^2}&k^{q^2}&1\\
        1&k^{q^3}&1&m^{q^3}&k^{q^3}\\
        k^{q^4}&1&k^{q^4}&1&m^{q^4}
    \end{pmatrix}.
\end{equation*}
We have that \eqref{rankZk} holds if and only if $\det(M)=0$ and $\det(M^{2,3,4,5}_{1,2,3,4})=0$ \cite{MR4038645}. These conditions, together with their Frobenius conjugates, yield six equations in the variables $m_i \in \F_q$, where $m = \sum_{i=0}^4 m_i \xi^{q^i}$ and $\{\xi, \xi^q, \dots, \xi^{q^4}\}$ is a normal basis of $\F_{q^5}$ over $\F_q$. These equations define a variety over $\F_q$, parametrized by $k \in \F_{q^5}$. Our aim is to prove, for each $k$ satisfying our constraints, the existence of a solution $(m_0, m_1, m_2, m_3, m_4) \in \F_q^5$ to these six equations.

Our proof is divided into several steps.
\begin{enumerate}
    \item First, we show that there is a smaller set of equations, $\{h_1, h_2, h_3, h_4\}$, whose solutions are also solutions to the original system in terms of $(m_0, m_1, m_2, m_3, m_4) \in \F_q^5$.
    
    \item Let $\mathcal{V}$ be the variety defined by $\{h_1, h_2, h_3, h_4\}$. We consider the affine transformation
   \begin{equation}\label{Eq:phi}
        \phi(m_0,m_1,m_2,m_3,m_4)=\Big(\sum_{i=0}^4 \xi^{q^i} m_i, \dots, \sum_{i=0}^4 \xi^{q^{i+4}}m_{i+4 \pmod 5}\Big),
    \end{equation}
    and the corresponding variety $\mathcal{W}$. This projectivity preserves the absolute irreducibility of varieties but is not defined over $\F_q$.
    
    \item We prove that $\mathcal{W}$ is absolutely irreducible for each $k$ satisfying our constraints, and thus so is $\mathcal{V}$. Furthermore, $\mathcal{W}$ is fixed by the shift $\psi$ that sends each $m_i$ to $m_{i+1 \pmod 5}$ and raises the coefficients to the power $q$. This shows that $\mathcal{V}$ is $\F_q$-rational (i.e., fixed by the Frobenius map $\phi_q$).
\end{enumerate}
    
First, we consider the following set of equations $\{h_1(m),h_2(m),h_3(m),h_4(m)\}$, where 

\begin{eqnarray*}
    h_1(m) &:=&2 m_{q^4+q+1} + k^{q^4+q^3}m^{q+1} + k^{q^3+q}m^{q^4+1}+ 2 k^{q^4+q^3+q}m + 2 k^{q^3+q}m  \\
    &&+k^qm+ m^{q^4+q}+ 2 k^{q^4+q^3}m^q + 
    2 k^{q+1}m^{q^4} + k m^{q^4} + 2k^{q^3} m^{q^4} \\
    &&+ k^{q^4+q^3+q+1} + 2k^{q^4+q^3+1}+ k^{q^4+q^3} + k^{q^3} + 2;\\
    h_2(m) &:=&2 m^{q^3+q^2+q}+ k^{q^3+1}m^{q^2+q} + m^{q^3+q} + 2km^{q} + 2 k^{q^3+q^2}m^{q} + k^{q^2} m^q\\ &&+ k^{q+1}m^{q^3+q^2} + 2 k^{q^3+q^2+1}m^{q^2} + 
    2 k^{q^3+1}m^{q^2}+ k^{q^3}m^{q^2}\\ 
    &&+ 2k^{q+1} m^{q^3} + k^{q^3+q^2+q+1}+ 2 k^{q^2+q+1} + k^{q+1} + k + 2;\\
    h_3(m) &:=& 2 ^{q^2+q+1} + k^{q^4+q^2} m^{q+1} + m^{q^2+1} + 2k^{q^2+q} m + k^q m + 2 k^{q^4}m\\ 
    &&+ k^{q^4+1}m^{q^2+q} + 2 k^{q^4+q^2+1}m^q + 
    2k^{q^4+q^2} m^q+ k^{q^2}m^q + 2k^{q^4+1} m^{q^2}\\
    &&+ k^{q^4+q^2+q+1} + 2 k^{q^4+q+1} + k^{q^4+1}+ k^{q^4}+ 2;\\
    h_4(m) &:=& (k^{q^4+q^3+q^2} + 2) m^{2q+2}
    + (2  k^{q^4+q^3+q^2+q} + 2 k^{q^4+q^3+q^2} + 
    k^{q^3+q} +  k^{q} + 2k^{q^4+q^3} +1)m^{2q+1}\\
    &&+ (k^{q^3+q^2+2q} +  2 k^{q^3+2q} + k^{q^4+q^3+q} + 2 k^q)m^2 + (2 k^{q^4+q^3+q^2+1}+  k^{q^4+1} + 2k^{q^4+q^3+q^2}\\
    &&+ 2 k^{q^4+q^2} + k^{q^2} + 1)m^{q+2} 
    + (2  k^{q^4+q^3+q^2+q+1} + k^{q^3+q^2+q+1} + k^{q^4+q^2+q+1}+ 2 k^{q^2+q+1}\\
    &&+ 2k^{q^4+q^3+q+1} + 2 k^{q^4+q+1} + 2 k^{q+1} + 
    2 k^{q^4+q^3+q^2+1} + k^{q^4+q^3+1}+ 2 k^{q^4+1} +  k\\
    &&+ k^{q^4+q^3+q^2+q}+ 2 k^{q^3+q^2+q}+  k^{q^2+q} + 2 k^q+ k^{q^4+q^3+q^2}+  k^{q^3+q^2} + 2 k^{q^2} +  k^{q^4+q^3} + 2 k^{q^3}\\
    &&+k^{q^4} + 1)m^{q+1} + (2 k^{q^4+q^3+q^2+2q+1}+2 k^{q^4+q^3} + 2 k^{q^2+2q+1} + 
    k^{q^4+q^3+2q+1} + k^{2q+1}\\
    &&+ k^{q^2+q+1} + 2 k^{q^4+q^3+q+1} +    2 k^{q^3+q+1} + k^{q+1} + k^{q^4+q^3+1}+ 2 k + 2k^{q^3+q^2+q} + 
    2 k^{q^4+q^3+q}\\
    &&+ 2k^{q^3+q} + 2 k^{q^4+q^3} + 1)m + ( k^{q^4+q^3+q^2+1} + 
    2 k^{q^4+1} +  k^{q^4+q^2} + 2 k^{q^2})m^{2q}\\
    &&+ (2 k^{q^4+q^3+q^2+q+2} + k^{q^4+q+2} + k^{q^4+q^3+q^2+2} + 2 k^{q^4+2}+ k^{q^4+q^2+q+1}+ 
    k^{q^2+q+1} + k^{q^4+q+1}\\
    &&+ 2 k^{q^4+q^3+q^2+1} + 2 k^{q^3+q^2+1}+ k^{q^4+q^2+1} + k^{q^4+1} + 2 k^{q^4+q^3+q^2} + k^{q^2+q^3} + 2 k^{q^4}+ 1)m^q\\
    &&+ k^{q^4+q^2+2q+2}+ 
    2 k^{q^4+2q+2} + 2 k^{q^4+q^2+q+2} + k^{q^4+q^3+q+2} + 2 k^{q^4+q^3+2} + k^{q^4+2}\\
    &&+ k^{q^4+q^3+q^2+q+1}+ 2 k^{q^4+q^3+q+1} + k^{q^4+q+1} + 2 k^{q+1} + k^{q^3+1} + 2 k^{q^4+1}+ k^{q^4+q^3} + 2 k^{q^3}.
\end{eqnarray*}
A direct computation using MAGMA confirms that if $m$ is a common root of the four preceding polynomials, then the rank of the matrix $M$ is less than 4. This can be verified by noting that the polynomials $h_1, h_2,$ and $h_3$ are linear in $m^{q^4}, m^{q^3},$ and $m^{q^2}$, respectively. One can therefore solve for these terms and substitute the resulting expressions into the conditions $\det(M)=0$ and $\det(M^{2,3,4,5}_{1,2,3,4})=0$. After making these substitutions, the two resulting polynomials in $m$ and $m^q$ are found to be divisible by $h_4$. This accomplishes Task (a).

We now establish the existence of an element $m \in \mathbb{F}_{q^5}$ that is a common zero of the polynomials $h_1, h_2, h_3,$ and $h_4$.

To this end, we employ an approach based on algebraic varieties. We represent an element $m \in \mathbb{F}_{q^5}$ using a normal basis $\{\xi, \xi^q, \dots, \xi^{q^4}\}$, such that $m = \sum_{i=0}^4 m_i \xi^{q^i}$ with coordinates $m_i \in \mathbb{F}_q$. The four equations $h_j(m_0, m_1, m_2, m_3, m_4)=0$ for $j=1,2,3,4$ define a variety $\mathcal{V} \subseteq \mathbb{A}^5(\mathbb{F}_{q^5})$ over $\mathbb{F}_{q^5}$. Let $\mathcal{W}=\phi(\mathcal{V})$, where $\phi$ is defined in \eqref{Eq:phi}.

The Frobenius automorphism $\Phi_q: x \mapsto x^q$ acts on an element $m$ by mapping it to $\sum_{i=0}^4 m_i \xi^{q^{i+1}}$, which induces a cyclic permutation on its coordinates. Let $\psi$ be the corresponding action on the polynomial ring, which raises the coefficients of polynomials to the $q$-th power and cyclically permutes the variables $m_i$.

To prove that the variety $\mathcal{V}$ is defined over $\mathbb{F}_q$, we must show that the defining ideal of $\mathcal{W}=\phi(\mathcal{V})$ is invariant under the action of $\psi$. 

Let $\widetilde{f_i} = h_i(\phi(m_0, m_1, m_2, m_3, m_4))$. To achieve our goal, it is sufficient to verify the following conditions:
\begin{itemize}
    \item $\psi(\widetilde{h_1}) = \widetilde{h_3}$
    \item $\psi(\widetilde{h_3}) = \widetilde{h_2}$
    \item The polynomials $\psi(\widetilde{h_2})$ and $\psi(\widetilde{h_4})$ are multiples of $\widetilde{h_4}$ modulo the ideal $\langle \widetilde{h_1}, \widetilde{h_2}, \widetilde{h_3} \rangle$.
\end{itemize}

The above conditions can be easily verified by using MAGMA.

The next step is to prove that $\mathcal{W}$ is absolutely irreducible. We begin by noting that the equations $\widetilde{h_1}=0, \widetilde{h_2}=0,$ and $\widetilde{h_3}=0$ are of degree one in the variables $m_2, m_3,$ and $m_4$. This linearity allows us to eliminate these variables, thereby establishing a birational equivalence between $\mathcal{W}$ and a plane curve defined over $\mathbb{F}_{q^5}$ in the variables $m_0$ and $m_1$. Note that $\widetilde{h_4}(m_0,m_1)$ can be written as $\sum_{i,j\leq 2}a_{i,j}m_0^im_1^j$, where 
\begin{eqnarray*}
    a_{2,2}&:=&k^{q^4+q^3+q^2} + 2;\\
    a_{2,1}&:=& 2  k^{q^4+q^3+q^2+q} + 2 k^{q^4+q^3+q^2} + 
    k^{q^3+q} +  k^{q} + 2k^{q^4+q^3} +1;\\
    a_{2,0}&:=& k^{q^3+q^2+2q} +  2 k^{q^3+2q} + k^{q^4+q^3+q} + 2 k^q;\\
    a_{1,2}&:=&2 k^{q^4+q^3+q^2+1}+  k^{q^4+1} + 2k^{q^4+q^3+q^2}+ 2 k^{q^4+q^2} + k^{q^2} + 1;\\
    a_{1,1}&:=&2  k^{q^4+q^3+q^2+q+1} + k^{q^3+q^2+q+1} + k^{q^4+q^2+q+1}+ 2 k^{q^2+q+1}+ 2k^{q^4+q^3+q+1}+ 2 k^{q^4+q+1}\\
    &&  + 2 k^{q+1} + 
    2 k^{q^4+q^3+q^2+1} + k^{q^4+q^3+1}+ 2 k^{q^4+1} +  k+ k^{q^4+q^3+q^2+q}+ 2 k^{q^3+q^2+q}\\
    &&+  k^{q^2+q} + 2 k^q+ k^{q^4+q^3+q^2}+  k^{q^3+q^2} + 2 k^{q^2} +  k^{q^4+q^3} + 2 k^{q^3}+k^{q^4} + 1;\\
    a_{1,0}&:=&2 k^{q^4+q^3+q^2+2q+1}+2 k^{q^4+q^3} + 2 k^{q^2+2q+1} + 
    k^{q^4+q^3+2q+1} + k^{2q+1}+ k^{q^2+q+1} \\
    &&+ 2 k^{q^4+q^3+q+1} +    2 k^{q^3+q+1} + k^{q+1} + k^{q^4+q^3+1}+ 2 k + 2k^{q^3+q^2+q} + 
    2 k^{q^4+q^3+q}\\
    &&+ 2k^{q^3+q} + 2 k^{q^4+q^3} + 1;\\
    a_{0,2}&:=&k^{q^4+q^3+q^2+1} + 
    2 k^{q^4+1} +  k^{q^4+q^2} + 2 k^{q^2};\\
    a_{0,1}&:=& 2 k^{q^4+q^3+q^2+q+2} + k^{q^4+q+2} + k^{q^4+q^3+q^2+2} + 2 k^{q^4+2}+ k^{q^4+q^2+q+1}+ 
    k^{q^2+q+1} + k^{q^4+q+1}\\
    &&+ 2 k^{q^4+q^3+q^2+1} + 2 k^{q^3+q^2+1}+ k^{q^4+q^2+1} + k^{q^4+1} + 2 k^{q^4+q^3+q^2} + k^{q^2+q^3} + 2 k^{q^4}+ 1;\\
    a_{0,0}&:=& k^{q^4+q^2+2q+2}+ 
    2 k^{q^4+2q+2} + 2 k^{q^4+q^2+q+2} + k^{q^4+q^3+q+2} + 2 k^{q^4+q^3+2} + k^{q^4+2}\\
    &&+ k^{q^4+q^3+q^2+q+1}+ 2 k^{q^4+q^3+q+1} + k^{q^4+q+1} + 2 k^{q+1} + k^{q^3+1} + 2 k^{q^4+1}+ k^{q^4+q^3} + 2 k^{q^3}.
\end{eqnarray*}

Suppose that $(m_1-\lambda)$ is a factor of $\widetilde{f_4}(m_0,m_1)$ for some $\lambda\in \overline{\mathbb{F}_q}$. Then $\widetilde{h_4}(m_0,\lambda)$ must be the zero polynomial in $m_0$. By inspecting the coefficient of $m_0^2$, we deduce that either
$$\lambda=k^q \quad \textrm{or} \quad \lambda=\frac{k^{q^3+q^2+q}+ 2k^{q^3+q}+ k^{q^4+q^3}+ 2}{k^{q^4+q^3+q^2} + 2}.$$
Note that, since $\NN_{q^5/q}(k)=1$, the denominator above vanishes only if $k=1$.

\begin{itemize}
    \item In the former case, we substitute $m_1=k^q$ into $\widetilde{f_4}(m_0,m_1)$. Setting the coefficient of $m_0$ to zero yields
    $$(k^{q^3+q}+ k^{q^4+q}+ k^q + 2k^{q^4+q^3} + 1)(k^{q^2+q} + 2)(k + 2)=0.$$
    The last two factors yield $k=1$, so we are left with the condition $k^{q^3+q}+ k^{q^4+q}+ k^q + 2k^{q^4+q^3} + 1=0$. It can be verified that no $k\in \mathbb{F}_{q^5}$ satisfying $\NN_{q^5/q}(k)=1$, $k\neq 1$, $u_1(k)u_2(k)=0$, and $t(k)=0$ also satisfies this last condition.
    
    \item In the latter case, we again substitute the expression for $m_1$ into $\widetilde{h_4}(m_0,m_1)$ and consider the coefficient of $m_0$. Setting this coefficient to zero implies that either $k=1$ or
    \begin{multline*}
     k^{q^{4}}+k^{1+2q+2q^{2}+3q^{3}+3q^{4}}+ k^{1+2q+2q^{2}+3q^{3}+2q^{4}}+ k^{1+2q+2q^{2}+2q^{3}+q^{4}}+ k^{1+2q+q^{2}+3q^{3}+2q^{4}}\\
     +2 k^{q^{3}}+2 k^{1+2q+q^{2}+3q^{3}+q^{4}}+2 k^{1+2q+q^{2}+2q^{3}+2q^{4}}+2 k^{1+2q+q^{2}+2q^{3}+q^{4}}+ k^{1+2q+q^{2}+2q^{3}}\\
     +2 k^{1+2q+q^{2}+q^{3}+q^{4}}+2 k^{1+2q+q^{2}+q^{3}}+2 k^{1+2q+q^{3}}+ k^{1+2q}+ k^{1+q+2q^{2}+3q^{3}+2q^{4}}\\
     +2 k^{1+q+2q^{2}+2q^{3}+2q^{4}}
     +2 k^{1+q+q^{2}+3q^{3}+3q^{4}}+ k^{1+q+q^{2}+3q^{3}+2q^{4}}+2 k^{1+q+q^{2}+2q^{3}+3q^{4}}\\
     +2 k^{1+q+q^{2}+2q^{3}+2q^{4}}
     + k^{1+q+q^{2}+q^{3}+2q^{4}}+ k^{1+q+q^{2}+q^{3}+q^{4}}+2 k^{1+q+2q^{3}+2q^{4}}+ k^{1+q+2q^{3}+q^{4}}\\
     + k^{1+q+q^{3}+2q^{4}}
     +2 k^{1+q+q^{4}}+ k^{1+2q^{3}+3q^{4}}+2 k^{1+2q^{3}+2q^{4}}
     +2 k^{1+q^{3}+2q^{4}}+ k^{1+q^{3}+q^{4}}\\+2 k^{2q+2q^{2}+3q^{3}+2q^{4}}+2 k^{2q+2q^{2}+3q^{3}+q^{4}}
     +2 k^{2q+2q^{2}+2q^{3}}+2 k^{2q+q^{2}+3q^{3}+q^{4}}+ k^{2q+q^{2}+2q^{3}+q^{4}}\\+2 k^{2q+q^{2}+2q^{3}}
     + k^{2q+q^{2}+q^{3}}+ k^{2q+2q^{3}}+2 k^{2q+q^{3}}+2 k^{q+2q^{2}+3q^{3}+q^{4}}
     + k^{q+2q^{2}+2q^{3}+q^{4}}\\+ k^{q+q^{2}+3q^{3}+2q^{4}}+ k^{q+q^{2}+3q^{3}+q^{4}}+ k^{q+q^{2}+2q^{3}+2q^{4}}
     +2 k^{q+q^{2}+2q^{3}+q^{4}}+ k^{q+q^{2}+2q^{3}}\\\ +2 k^{q+q^{2}+q^{3}+q^{4}}+ k^{q+q^{2}+q^{3}}
     +2 k^{q+2q^{3}+q^{4}}+2 k^{q+2q^{3}}+ k^{q+q^{3}+q^{4}}+2 k^{q+q^{3}}+2 k^{q}\\
     +2 k^{q^{2}+3q^{3}+2q^{4}}+ k^{q^{2}+2q^{3}+2q^{4}}+ k^{q^{2}+2q^{3}+q^{4}}+2 k^{q^{2}+q^{3}+q^{4}}
     + k^{2q^{3}+q^{4}}+2 k^{q^{3}+2q^{4}} =0.
    \end{multline*}
    As in the previous case, it can be verified that no $k\in \mathbb{F}_{q^5}$ satisfying $\NN_{q^5/q}(k)=1$ and $k\neq 1$, along with $u_1(k)u_2(k)=0$ and $t(k)=0$, satisfies this final equation.
\end{itemize}
Since a similar argument holds for factors of the form $(m_0-\lambda)$, where $\lambda \in \overline{\mathbb{F}_q}$, we conclude that $\widetilde{h_4}(m_0,m_1)$ has no linear (degree-one) factors. Therefore, if $\widetilde{h_4}(m_0,m_1)$ is not absolutely irreducible, it must split into two quadratic (degree-two) factors. Furthermore, both quadratic factors must be linear in each variable, $m_0$ and $m_1$. Otherwise, a factor depending on only one variable would have to exist, which contradicts the previous conclusion.
Thus, the only possibility is that the two factors are
\begin{align*}
j_1(m_0,m_1) &:= (m_1 + 2 k^q)m_0+A m_1+B, \\
j_2(m_0,m_1) &:= ((k^{q^4+q^3+q^2} + 2)m_1+2k^{q^3+q^2+q}+ k^{q^3+q}+ 2k^{q^4+q^3}+ 1 ) m_0+C m_1+D,
\end{align*}
for some $A,B,C,D \in \overline{\mathbb{F}_{q}}$.

The polynomial $J(m_0,m_1):=\widetilde{h_4}(m_0,m_1)-j_1(m_0,m_1)j_2(m_0,m_1)$ must vanish identically. By setting the coefficients of $m_0$ and $m_0m_1^2$ in $J(m_0,m_1)$ to zero, we get:
\begin{align*}
-k^q D ={}& 2k^{1+2q+q^{2}+q^{3}+q^{4}}+2k^{1+2q+q^{2}}+ k^{1+2q+q^{3}+q^{4}}+ k^{1+2q}+ k^{1+q+q^{2}} \\
& +2k^{1+q+q^{3}+q^{4}} +2k^{1+q+q^{3}}+ k^{1+q}+ k^{1+q^{3}+q^{4}}+2k+2k^{q+q^{2}+q^{3}} \\
& +2k^{q+q^{3}+q^{4}} +2k^{q+q^{3}} +2k^{q^{3}+q^{4}}+1+(k^{q^3+q^2+q} + 2+k^{q^4+q^3}+2k^{q^3+q})B, \\
C ={}& (2k^{q^4+q^3+q^2}+1)A +2k^{q^4+q^3+q^2+1}+k^{q^4+1} + 2k^{q^4+q^3+q^2} + 2k^{q^4+q^2} + k^{q^2} +1.
\end{align*}
From the coefficients of $m_1^2$ and $m_0m_1$, we also obtain the conditions:
\begin{equation*}
(A + 1)(k^{q^4+q^3+q^2+1}+ 2k^{q^4+1} + k^{q^4+q^3+q^2}A + k^{q^4+q^2}+ 2k^{q^2} + 2 A)=0,
\end{equation*}
and
\begin{multline*}
(k^{q+q^{2}+q^{3}+q^{4}}+2k^{q+q^{2}+q^{3}}+ k^{q+q^{3}}+2k^{q}+2k^{q^{3}+q^{4}}+1)B \\
+(k^{2q+q^{2}+q^{3}+q^{4}}+2k^{2q+q^{2}+q^{3}}+ k^{2q+q^{3}}+2k^{2q}+2k^{q+q^{3}+q^{4}} + k^{q})A \\
+2k^{1+2q+q^{2}+q^{3}}+2k^{1+2q+q^{2}+q^{4}}+2k^{1+2q+q^{2}} + k^{1+q+q^{2}+q^{3}+q^{4}}+2k^{1+q+q^{2}} \\
+ k^{1+q+q^{3}} + k^{1+q+q^{4}}+ k^{1+q}+2k^{1+q^{3}+q^{4}}+ 2k+ k^{2q+q^{2}+q^{3}}+ k^{2q+q^{2}+q^{4}}+ k^{2q+q^{2}} \\
+2k^{q+q^{2}+q^{3}+q^{4}}+ k^{q+q^{2}}+2k^{q+q^{3}}+2k^{q+q^{4}}+2k^{q}+ k^{q^{3}+q^{4}}+2 =0.
\end{multline*}
The equation involving $A$ yields two possibilities. Either 
$$A=-1 \textrm{ or } A = - \frac{k^{q^4+q^3+q^2+1}+ 2k^{q^4+1} + k^{q^4+q^2}+ 2k^{q^2}}{k^{q^4+q^3+q^2} + 2}.$$ It can be verified that both cases lead to a contradiction when substituted into the other coefficient equations, given the conditions on $k$ (namely that $\NN_{q^5/q}(k)=1$, $k\neq 1$, $u_1(k)u_2(k)=0$, and $t(k)=0$).

This shows that the curve $\mathcal{V}$ is $\mathbb{F}_q$-rational and absolutely irreducible. An upper bound on its genus can be determined by examining the plane curve defined by $\widetilde{h_4}(m_0,m_1)=0$. This curve has degree four and at least two singular points (its points at infinity), so its geometric genus is at most one. Since the genus is a birational invariant (i.e., it is preserved by such transformations), and both $\mathcal{V}$ and $\mathcal{W}$ are birationally equivalent to this plane curve, we conclude that the genus of both $\mathcal{V}$ and $\mathcal{W}$ is at most one. Therefore, by the celebrated Hasse-Weil Theorem \cite{doi:10.1073/pnas.27.7.345}, the curve $\mathcal{V}$ has at least one affine $\mathbb{F}_q$-rational point (recalling that $q=3^{2h+1}$ and $h>1$). This implies that $\ZZ_k$ is not scattered, which concludes the proof.
\end{proof}

We can now prove our main result.

\begin{thm}\label{MainTh}
Let $q=3^{2h+1}$, with $h>1$. The subspace
$$ U_{1} := \{(x, x^{q} + x^{q^4}) : \Tr_{q^5/q}(x)=0\} \leq \mathbb{F}_{q^5}^2 $$
is maximally scattered.
\end{thm}

\begin{proof}
This directly follows from Propositions \ref{Prop:Cond}, \ref{Prop:C1}, \ref{Prop:C2}, \ref{Prop:C3}, and \ref{Prop:C4}.
\end{proof}

\section{Self-duality of the code $\mathcal{C}_\delta$}\label{Sec:DualCode}

Let $\mathcal{C}$ be an $[n,k]_{q^m/q}$ rank-metric code. Its dual code $\mathcal{C}^\perp$ is the $[n,n-k]_{q^m/q}$ code defined by the standard inner product on $\F_{q^m}^n$. It is a well-known geometric property (see \cite{BorelloPolverinoZullo2025,csajbok2021generalising}) that if $U$ is a $q$-system associated with $\mathcal{C}$, then the $q$-system associated with the dual code $\mathcal{C}^\perp$ is exactly the \emph{Delsarte dual} $\bar{U}$ of $U$. 

In our specific case, $U_\delta$ is a 4-dimensional $\F_q$-subspace of $\mathbb{F}_{q^5}^2$. The associated code $\mathcal{C}_\delta$ is a $[4,2]_{q^5/q}$ code, and its dual $\mathcal{C}_\delta^\perp$ is also a $[4,2]_{q^5/q}$ code. Computing the Delsarte dual $\bar{U}_\delta$ reduces directly to determining the generators of $\mathcal{C}_\delta^\perp$.

\begin{thm}\label{Thm:SelfDual}
Let $q$ be an odd prime power, $\delta \in \F_{q^5}$ with $\NN_{q^5/q}(\delta) = 1$. Let $\mathcal{C}_\delta$ be the $[4,2]_{q^5/q}$ code associated with the $q$-system
$$U_\delta := \{(x,\, x\q + \delta x\qqqq) : \Tr_{q^5/q}(x) = 0\} \subseteq \mathbb{F}_{q^5}^2.$$
Then $\mathcal{C}_\delta^\perp$ is equivalent to $\mathcal{C}_\delta$. Geometrically, this means $U_\delta$ is Delsarte self-dual.
\end{thm}

In order to prove Theorem \ref{Thm:SelfDual}, we need a few auxiliary results.

\medskip 

Fix an $\F_q$-basis $x = (x_1, x_2, x_3, x_4)$ of $\ker(\Tr_{q^5/q})$ and write $x^{(q^i)} = (x_1^{q^i}, \ldots, x_4^{q^i})$ for $i\in\{1,\ldots,4\}$. The code associated with $U_\delta$ is
$$\mathcal{C}_\delta =\langle x,\;\; x^{(q)} + \delta\,x^{(q^4)}\rangle_{\F_{q^5}} \subseteq \F_{q^5}^4.$$
For $c= (c_1, c_2, c_3, c_4) \in \F_{q^5}^4$, set $\sigma_i(c) := \sum_{j=1}^{4} x_j^{q^i} c_j$. Then $c \in \mathcal{C}_\delta^\perp$ if and only if
\begin{equation}\label{eq:dual_cond}
\sigma_0(c) = 0 \qquad\text{and}\qquad \sigma_1(c) + \delta\,\sigma_4(c) = 0.
\end{equation}
Since $x_j \in \ker(\Tr_{q^5/q})$, we have $x_j\qqqq = -x_j - x_j\q - x_j\qq - x_j\qqq$, giving
\begin{equation}\label{eq:sigma4_rel}
\sigma_4(c) = -\sigma_0(c) - \sigma_1(c) - \sigma_2(c) - \sigma_3(c).
\end{equation}

Consider the $4 \times 4$ Moore matrix
$$M_4(x) = \begin{pmatrix} x_1 & x_2 & x_3 & x_4 \\ x_1\q & x_2\q & x_3\q & x_4\q \\ x_1\qq & x_2\qq & x_3\qq & x_4\qq \\ x_1\qqq & x_2\qqq & x_3\qqq & x_4\qqq \end{pmatrix}.$$
 Since $x_1, \ldots, x_4$ are $\F_q$-independent in $\F_{q^5}$, the matrix $M_3(x)$ has rank $3$ and $M_4(x)$ is non-singular. Set $s := \det(M_4(x))  \neq 0$ and define $y = (y_1, y_2, y_3, y_4)$ as the vector of cofactors of the last row of $M_4(x)$:
$$y_j := (-1)^{j}\,\det(M_3(x_1, \ldots, \widehat{x_j}, \ldots, x_4)), \qquad j = 1, 2, 3, 4,$$
where $(x_1, \ldots, \widehat{x_j}, \ldots, x_4)$ denotes the vector without the $j$-th component. By the standard properties of cofactor expansions, $\sigma_i(y) = 0$ for $i = 0, 1, 2$ (expansion along a wrong row produces a matrix with two identical rows) and $\sigma_3(y) = s$ (expansion along the correct row). Together with~\eqref{eq:sigma4_rel}:
\begin{equation}\label{eq:sigma_vals}
\sigma_0 = \sigma_1 = \sigma_2 = 0, \qquad \sigma_3 = s, \qquad \sigma_4 = -s.
\end{equation}

\begin{lemma}\label{Lem:SigmaIK}
Setting $\sigma_{ik} := \sum_j x_j^{q^i} y_j^{q^k} = \sigma_i(y^{(q^k)})$, one has $\sigma_{ik} = \sigma_{(i-k)}(y)^{q^k}$.
\end{lemma}
\begin{proof}
$\sigma_{ik} = \bigl(\sum_j x_j^{q^{i-k}} y_j\bigr)^{q^k}$.
\end{proof}

Combined with~\eqref{eq:sigma_vals}, the complete table of $\sigma_i(y^{(q^k)})$ is:

\medskip
\begin{center}
\renewcommand{\arraystretch}{1.3}
\begin{tabular}{c|ccc|c}
$k$ & $\sigma_0$ & $\sigma_1$ & $\sigma_4$ & $\sigma_1 + \delta\,\sigma_4$ \\[2pt]
\hline
$0$ & $0$ & $0$ & $-s$ & $-\delta\,s$ \\
$1$ & $-s\q$ & $0$ & $s\q$ & $\delta\,s\q$ \\
$2$ & $s\qq$ & $-s\qq$ & $0$ & $-s\qq$ \\
$3$ & $0$ & $s\qqq$ & $0$ & $s\qqq$ \\
$4$ & $0$ & $0$ & $0$ & $0$ \\
\end{tabular}
\end{center}

\medskip

The row $k = 4$ shows that $y^{(q^4)} \in \mathcal{C}_\delta^\perp$. For the second generator, set $c_2 = y + \alpha\,y^{(q^3)}$, with $\alpha\in\F_{q^5}$. The condition $\sigma_0(c_2) = 0 + \alpha \cdot 0 = 0$ is automatic, while condition~\eqref{eq:dual_cond} gives
$$\alpha\,s\qqq + \delta(-s) = 0 \qquad\Longrightarrow\qquad \alpha = \delta\,s^{1-q^3}.$$
The two generators are $\F_{q^5}$-independent since they involve distinct Frobenius shifts of $y$, so
\begin{equation}\label{eq:Cperp_gen}
\mathcal{C}_\delta^\perp = \langle y^{(q^4)},\;\; y + \delta\,s^{1-q^3}\,y^{(q^3)}\rangle_{\F_{q^5}}.
\end{equation}

\begin{proposition}\label{Prop:MooreFq}
$s = \det(M_4(x)) \in \F_q^{\times}$.
\end{proposition}
\begin{proof}
Applying the $q$-th power to the determinant shifts each row by one Frobenius:
$$s\q =\det(M_4(x^{(q)}))= \det\begin{pmatrix} x_1\q & x_2\q & x_3\q & x_4\q \\ x_1\qq & x_2\qq & x_3\qq & x_4\qq \\ x_1\qqq & x_2\qqq & x_3\qqq & x_4\qqq \\ x_1\qqqq & x_2\qqqq & x_3\qqqq & x_4\qqqq\end{pmatrix}.$$
Substituting $x_j\qqqq = -(x_j + x_j\q + x_j\qq + x_j\qqq)$ in the last row and expanding by multilinearity, all terms with a repeated row vanish, leaving
$$s\q = -\det\begin{pmatrix} x_1\q & \cdots & x_4\q \\ x_1\qq & \cdots & x_4\qq \\ x_1\qqq & \cdots & x_4\qqq \\ x_1 & \cdots & x_4\end{pmatrix} = (-1)(-1)^3\,s = s,$$
where three row swaps bring the last row to the top, restoring the standard Moore order.
\end{proof}

Since $s \in \F_q^{\times}$, we have $s^{q^3} = s$, hence $\alpha = \delta\,s^{1-q^3} = \delta$, and~\eqref{eq:Cperp_gen} simplifies to
\begin{equation}\label{eq:Cperp_final}
\mathcal{C}_\delta^\perp = \langle y^{(q^4)},\;\; y + \delta\,y^{(q^3)}\rangle_{\F_{q^5}}.
\end{equation}

\begin{proposition}\label{Prop:WkerTr}
Let $y = (y_1, y_2, y_3, y_4)$. The components of $y$ form an $\mathbb{F}_q$-basis of $\ker(\Tr_{q^5/q})$.
\end{proposition}
\begin{proof}
Set $v = \sum_{k=0}^{4}y^{(q^k)}$, so that $v_j = \Tr_{q^5/q}(y_j)$. Then
$$\sigma_i(v) = \sum_{k=0}^{4} \sigma_{ik} = s^{q^{(i-3)}} - s^{q^{(i-4)}} = s - s = 0$$
for all $i=1,\dots,4$, where the last equality uses $s \in \mathbb{F}_q$. Since the $4\times 4$ Moore matrix $M_4(x)$ is non-singular, $v$ must be the null vector. Hence $y_j \in \ker(\Tr_{q^5/q})$ for every $j$. Since $\det(M_4(y)) \neq 0$ (which follows from standard cofactor matrix properties), the elements $y_1, y_2, y_3, y_4$ are $\mathbb{F}_q$-linearly independent, forming a basis of $\ker(\Tr_{q^5/q})$.
\end{proof}

We are now ready to prove Theorem \ref{Thm:SelfDual}.

\begin{proof}[Proof of Theorem \ref{Thm:SelfDual}]
By \eqref{eq:Cperp_final}, the dual code $\mathcal{C}_\delta^\perp$ is generated by the rows of the matrix
$$H = \begin{pmatrix} y^{(q^4)} \\ y + \delta y^{(q^3)} \end{pmatrix}.$$
Let $u = y^{(q^4)}$. Notice that $\ker(\Tr_{q^5/q})$ is closed under the Frobenius map $z \mapsto z^q$, since $\Tr_{q^5/q}(z^q) = \Tr_{q^5/q}(z)^q = 0$. Because the map $z \mapsto z^{q^4}$ is an $\mathbb{F}_q$-automorphism of $\mathbb{F}_{q^5}$ that preserves this kernel, the components of $u$ also form an $\mathbb{F}_q$-basis of $\ker(\Tr_{q^5/q})$. 

Note that $u^{(q)} = y^{(q^5)} = y$ and $u^{(q^4)} = y^{(q^8)} = y^{(q^3)}$. Thus, the generator matrix $H$ can be rewritten as
$$H = \begin{pmatrix} u \\ u^{(q)} + \delta u^{(q^4)} \end{pmatrix}.$$
Recall that the generator matrix of the original code $\mathcal{C}_\delta$ is
$$G = \begin{pmatrix} x \\ x^{(q)} + \delta x^{(q^4)} \end{pmatrix}.$$
Since both $x = (x_1, x_2, x_3, x_4)$ and $u = (u_1, u_2, u_3, u_4)$ are $\mathbb{F}_q$-bases of the same 4-dimensional $\mathbb{F}_q$-vector space $\ker(\Tr_{q^5/q})$, there exists an invertible matrix $A \in \GL(4, q)$ such that
$$u = x A.$$
Because $A$ has entries in $\mathbb{F}_q$, the Frobenius map acts component-wise and leaves $A$ invariant. Therefore, $u^{(q)} = x^{(q)} A$ and $u^{(q^4)} = x^{(q^4)} A$. 
Substituting this into $H$, we obtain
$$H = \begin{pmatrix} x A \\ x^{(q)}A + \delta x^{(q^4)}A \end{pmatrix} = \begin{pmatrix} x \\ x^{(q)} + \delta x^{(q^4)} \end{pmatrix} A = G A.$$
Since two rank-metric codes are equivalent if they have two generator matrices differing by right multiplication by a matrix in $\GL(n, q)$, the codes $\mathcal{C}_\delta^\perp$ and $\mathcal{C}_\delta$ are equivalent.
\end{proof}

\section{Concluding remarks and open problems}
In this paper, we presented the first infinite family of non-extendable $\mathbb{F}_{q^m}$-linear MRD codes that do not reach the maximum possible length (i.e., $n < m$). Geometrically, these correspond to scattered subspaces that are not maximum scattered, but are maximal with respect to inclusion.

Our results rely heavily on the partial classification of maximum length MRD codes in $\mathbb{F}_{q^5}^2$ (corresponding to maximum scattered subspaces) obtained in \cite{Longobardi2025classificationmaximumscatteredlinear}. Consequently, if the codes corresponding to families (C3) and (C4) from Theorem~\ref{Th:Classification} are indeed non-existent---as conjectured and supported by computational evidence---then our construction holds for any value of $q$ and any $\delta$ satisfying the conditions in Proposition~\ref{Prop:Cond}.

We conclude with some open problems arising from this coding theory perspective:

\begin{enumerate}[(i)]
\item Are the MRD codes corresponding to families (C3) and (C4) in Theorem~\ref{Th:Classification} non-existent for any $q$?
\item Experimental results suggest that the codes associated with the $q$-systems $U_\delta$, $\NN_{q^m/q}(\delta)=1$, whenever they are MRD, are also non-extendable, at least when $m$ is odd. Is this true? Unfortunately, in these cases, we cannot use any partial classification of scattered subspaces in $\F_{q^m}^2$.
\item A general strategy for constructing non-extendable MRD codes could be as follows:
\begin{enumerate}
    \item Start with an $\mathbb{F}_{q^m}$-linear code $\mathcal{C}_f$ associated with $U_f := \{(x,f(x)) : x \in \mathbb{F}_{q^m}\}$, which is not MRD.
    \item Puncture the code by restricting its geometric domain to a suitable $(n-1)$-dimensional $\mathbb{F}_q$-subspace $H \subset \mathbb{F}_{q^m}$, creating a candidate MRD code $\mathcal{C}_{f,H}$ of length $n-1$.
\end{enumerate}
This approach raises a key question: For a given $f$ and $H$, can we determine all the possible  codes $\mathcal{C}_g$ that can act as one-column extensions of $\mathcal{C}_{f,H}$? Solving this would provide a powerful tool to prove non-extendability by ruling out known families of extended codes.

\item The discussion so far has focused exclusively on codes of dimension $2$ (scattered subspaces in $\mathbb{F}_{q^m}^2$). Until very recently, no examples were known of non-extendable MRD codes with higher dimension ($k\geq3$ or $h>1$). However, using MAGMA, we found a sporadic example of a non-extendable $[6,3,4]_{2^7/2}$ MRD code (geometrically, a maximally $2$-scattered subspace of dimension $6$ in $\mathbb{F}_{2^7}^3$). This isolated case suggests that such optimal, non-extendable codes may exist more generally.

\begin{rem}
A sporadic non-extendable $[6,3,4]_{2^7/2}$ MRD code can be described explicitly by its associated $q$-system:
$$
U = \{(x, x^2, x^{4}+\xi^3x^{64})+\mu(\xi^6, \xi^{27}, \xi^{34}) : x \in S,\mu\in\mathbb{F}_2\},
$$
where $
S:=\langle1,\xi,\xi^2,\xi^3,\xi^4\rangle_{\mathbb{F}_2}
$
and $\xi$ is a primitive element of $\mathbb{F}_{2^7}$.
It can be proven using MAGMA that $U$ is $2$-scattered but not contained in any larger $2$-scattered subspace, hence the associated MRD code cannot be extended.
\end{rem}

\item Are there structural differences between extendable and non-extendable MRD codes in terms of their generalized rank weights? Geometrically, this translates to investigating differences in the intersection numbers of maximally scattered versus non-maximally scattered subspaces with $\mathbb{F}_{q^m}$-subspaces.
\end{enumerate}
\section*{Acknowledgements}
The authors thank Bence Csajb{\'o}k for interesting and fruitful discussions. This research was partially supported by the Italian National Group for Algebraic and Geometric Structures and their Applications (GNSAGA—INdAM)
through the INdAM - GNSAGA Project ``\emph{Noncommutative polynomials in coding theory}'', CUP E53C24001950001.
This research was also partially supported by Bando Galileo 2026 - G26-260. 

\section*{Declarations}
{\bf Conflicts of interest.} The authors have no conflicts of interest to declare that are relevant to the content of this
article.

\bibliographystyle{abbrv} 
\bibliography{BIB} 

\end{document}